\documentclass[journal]{IEEEtran}

\usepackage{cite}

\usepackage{amsmath,amssymb,amsthm}
\usepackage{tikz}
\usetikzlibrary{arrows}
\usepackage{verbatim}
\usepackage{palatino}
\usepackage{mathpazo}
\usepackage{stmaryrd}
\usepackage{mathtools}
\mathtoolsset{centercolon}
\usepackage{thm-restate}
\newcommand\Tstrut{\rule{0pt}{2.6ex}}         

\usepackage{hyperref}
\usepackage{dsfont}

\theoremstyle{definition}
\newtheorem{corollary}{Corollary}
\newtheorem{definition}{Definition}

\newtheorem{lemma}{Lemma}
\newtheorem{proposition}{Proposition}
\newtheorem{theorem}{Theorem}

\newtheorem*{remark}{Remark}

\newcommand{\mbf}{\mathbf}
\newcommand{\mbb}{\mathbb}
\newcommand{\mc}{\mathcal}

\newcommand{\mK}{\mathcal{K}}
\newcommand{\mL}{\mathcal{L}}
\newcommand{\mM}{\mathcal{M}}
\newcommand{\mN}{\mathcal{N}}
\newcommand{\mS}{\mathcal{S}}
\newcommand{\mX}{\mathcal{X}}

\newcommand{\ol}{\overline}
\newcommand{\N}{\mathcal{N}}
\newcommand{\tr}{\textrm{Tr}}
\newcommand{\id}{\textrm{id}}
\newcommand{\ket}[1]{|#1\rangle}
\newcommand{\bra}[1]{\langle #1|}
\newcommand{\op}[2]{|#1\rangle\langle #2|}
\newcommand{\dyad}[1]{\op{#1}{#1}}
\newcommand{\pre}{\text{pre}}
\newcommand{\post}{\text{post}}

\newcommand{\cptp}{\text{CPTP}}

\newcommand{\wt}{\widetilde}
\newcommand{\cv}{\text{cv}}
\newcommand{\sep}{\text{SEP}}
\newcommand{\sepl}{\text{sep}}
\newcommand{\ppt}{\text{PPT}}
\newcommand{\sym}{\text{Sym}}
\newcommand{\Pos}{\mathrm{Pos}}
\newcommand{\T}{\textrm{T}}

\newcommand{\eb}{\text{EB}}

\newcommand{\opvec}{\operatorname{vec}}

\definecolor{cool_green}{rgb}{0.0, 0.5, 0.0}
\definecolor{navy_blue}{rgb}{0,0,0.5}

\newcommand{\todo}[1]{{\color{red} #1}}
\newcommand{\ian}[1]{{\color{navy_blue} #1}}

\begin{document}

\title{The Communication Value of a Quantum Channel}

\author{Eric Chitambar, Ian George, Brian Doolittle, Marius Junge%
\thanks{E. Chitambar and I. George are with the Department of Electrical and Computer Engineering, University of Illinois Urbana-Champaign, Urbana, IL, 61801 USA, e-mail: echitamb@illinois.edu.}
\thanks{B. Doolittle is with the Department of Physics, University of Illinois Urbana-Champaign, Urbana, IL, 61801 USA.}
\thanks{M. Junge is with the Department of Mathematics, University of Illinois Urbana-Champaign, Urbana, IL, 61801 USA.}}

\maketitle

\begin{abstract}
There are various ways to quantify the communication capabilities of a quantum channel.  In this work we study the communication value (cv) of channel, which describes the optimal success probability of transmitting a randomly selected classical message over the channel.  The cv also offers a dual interpretation as the classical communication cost for zero-error channel simulation using non-signaling resources.  We first provide an entropic characterization of the cv as a generalized conditional min-entropy over the cone of separable operators.  Additionally, the logarithm of a channel's cv is shown to be equivalent to its max-Holevo information, which can further be related to channel capacity.  We evaluate the cv exactly for all qubit channels and the Werner-Holevo family of channels.  While all classical channels are multiplicative under tensor product, this is no longer true for quantum channels in general.  We provide a family of qutrit channels for which the cv is non-multiplicative.  On the other hand, we prove that any pair of qubit channels have multiplicative cv when used in parallel.  Even stronger, all entanglement-breaking channels and the partially depolarizing channel are shown to have multiplicative cv when used in parallel with any channel.  We then turn to the entanglement-assisted cv and prove that it is equivalent to the conditional min-entropy of the Choi matrix of the channel.  Combining with previous work on zero-error channel simulation, this implies that the entanglement-assisted cv is the classical communication cost for perfectly simulating a channel using quantum non-signaling resources.  A final component of this work investigates relaxations of the channel cv to other cones such as the set of operators having a positive partial transpose (PPT).  The PPT cv is analytically and numerically investigated for well-known channels such as the Werner-Holevo family and the dephrasure family of channels.

\end{abstract}



\section{Introduction}

A noisy communication channel prohibits perfect transmission of messages from the sender (Alice) to the receiver (Bob).  While there are a number of ways to quantify the noise of a channel, perhaps the simplest is in terms of a guessing game.  Suppose that with uniform probability Alice randomly chooses a channel input and sends it to Bob over the channel.  Based on the channel output, Bob tries to guess Alice's input with the greatest probability of success.  In this game, Bob's optimal strategy is to perform maximum likelihood estimation based on the channel's transition probabilities.  To be concrete, suppose that $\mbf{P}:[n]\to[n']$ is a channel mapping set $[n]:=\{1,\cdots,n\}$ to set $[n']$ with transition probabilities $P(y|x)$.  We define the channel's \textit{communication value} (cv) to be
\begin{equation}
\label{Eq:Cv}
\cv(\mbf{P})=\sum_{y\in[n']}\max_{x\in[n]}P(y|x).
\end{equation}
It is then straightforward to see that $\frac{1}{n}\cv(\mbf{P})$ is the largest success probability of correctly identifying the input $x$  based on the output $y$, when $x$ is drawn uniformly from $[n]$.  The quantity $\cv(\mbf{P})$ is thus a natural measure for how well a channel $\mbf{P}$ transmits data on the single-copy level.  The goal of this paper is to better understand the channel cv in different communication settings.

The channel cv also emerges in the problem of zero-error channel simulation \cite{Cubitt-2011a, Duan-2016a, Wang-2016a, Fang-2020a}.  In the general task of channel simulation, Alice and Bob attempt to generate one channel $\mbf{P}$ using another channel $\mbf{Q}$ combined with pre- and post-processing \cite{Bennett-2002a, Bennett-2014a, Berta-2011a, Heinosaari-2019a, Heinosaari-2020a}.  Interesting variations to this problem arise when different types of resources are used to coordinate the pre- and post-processing of $\mbf{Q}$.  For example, these resources could be shared randomness \cite{Bennett-2014a, Frenkel-2015a}, shared quantum entanglement \cite{Berta-2013a, Wang-2018a, Wilde-2018a, Gour-2021a, Frenkel-2021a}, or non-signaling side-channel \cite{Cubitt-2011a, Duan-2016a, Fang-2020a}.  The latter refers to a general bipartite channel that prohibits communication from one party to the other.  When $\mbf{Q}=\id_r$ is the identity map on $[r]$, then the goal is to perfectly simulate $\mbf{P}$ using $r$ noiseless messages from Alice to Bob, along with any auxiliary resource.  For a given class of resource, the smallest number $r$ needed to accomplish this simulation is called the communication cost of $\mbf{P}$ (also referred to as the signaling dimension of $\mbf{P}$ in Refs. \cite{Dall'Arno-2017a, Doolittle-2021a}).  It turns out that $\lceil\cv(\mbf{P})\rceil$ is a lower bound on the communication cost when Alice and Bob have access to shared randomness \cite{Doolittle-2021a}.  In fact, this lower bound is tight when the Alice and Bob are allowed to use non-signaling resources \cite{Cubitt-2011a}.  Combining this discussion with the previous paragraph, we thus have two dual interpretations of the communication value: $\frac{1}{n}\cv(\mbf{P})$ as an optimal guessing probability and $\cv(\mbf{P})$ as an optimal simulation cost.  This is no coincidence since Eq. \eqref{Eq:Cv} is the dual formulation of the linear program characterizing the communication cost for perfectly simulating $\mbf{P}$ using classical non-signaling resources (see Section \ref{Sect:Capacity-NS} for more details).

The goal of this paper is to understand the communication value of quantum channels.  Formally, a quantum channel is described by a completely positive trace-preserving (CPTP) map $\mc{N}$ mapping density operators $\rho^A$ on Hilbert space $\mc{H}^A$ to density operators $\mc{N}(\rho)$ on Hilbert space $\mc{H}^B$.  Every quantum channel is able to generate a family of classical channels by encoding classical data into quantum objects.  Namely, for each $x\in[n]$, Alice prepares a quantum state $\rho_x$ and sends it through the channel to Bob's side.  Upon receiving $\mc{N}(\rho_x)$, Bob performs a quantum measurement, described by a general positive operator-valued measure (POVM) $\{\Pi_y\}_{y\in[n']}$, and regards his measurement outcome as the decoded classical data.  The induced classical channel then has the form
\begin{equation}
P(y|x)=\tr[\Pi_y\mc{N}(\rho_x)].
\end{equation}
How noisy this channel will be depends on the state encoding $\{\rho_x\}_{x\in[n]}$ and measurement decoding $\{\Pi_y\}_{y\in[n']}$, and ideally one chooses the states and measurement to minimize the error in data transmission.  
We define the cv of $\mN$ in terms of the classical channels it can generate.
\begin{definition}
\label{Defn:cv}
The $[n]\to[n']$ communication value (cv) of a quantum channel $\mc{N}\in\cptp(A\to B)$ is
\begin{equation}
\cv^{n\to n'}(\mc{N})=\max_{\substack{\{\Pi_{y}\}_{y=1}^{n'}\\\{\rho_x\}_{x=1}^n}}\{\cv(\mbf{P})\;|\; P(y|x)=\tr[\Pi_y\mc{N}(\rho_x)]\},
\end{equation}
and the cv value of $\mc{N}$ is defined as
\begin{equation}
\cv(\mc{N})=\sup_{n,n\in\mbb{N}}\cv^{n\to n'}(\mc{N}).
\end{equation}
\end{definition}
\noindent Analogous to the classical case, $\frac{1}{n}\cv^{n\to n'}(\mc{N})$ quantifies the largest success probability attainable in an $n$-input guessing game using the channel $\mc{N}$.  The quantity $\cv(\mc{N})$ also has a dual interpretation as the classical communication cost for simulating any classical channel generated by $\mc{N}$ when Alice and Bob have access to non-signaling resources.

By taking multiple copies of the channel, one can consider the cv capacity, defined as 
\begin{align}
\mc{CV}(\mbf{P})&=\lim_{k\to\infty}\frac{1}{k}\log\cv(\mbf{P}^{k}),\notag\\
\mc{CV}(\mN)&=\lim_{k\to\infty}\frac{1}{k}\log\cv(\mN^{\otimes k})
\end{align}
in the classical and quantum cases, respectively.  It is not difficult to see that $\log \cv(\mbf{P})$ is an additive quantity and so $\mc{CV}(\mbf{P})=\log\cv(\mbf{P})$.  On the other hand, as we show below, $\log\cv(\mN)$ is non-additive in general for quantum channels.  A primary objective of this paper is to understand when additivity of $\log\cv(\mN)$ (equivalently multiplicativity of $\cv(\mN)$) holds and when it does not.  One of our main results is that multiplicativity always holds for qubit channels, whereas it does not for qutrits.

One can also study the cv of channels that are enhanced by auxiliary resources shared between the sender and receiver.  In the quantum setting, it is natural to consider entanglement-assisted channel communication, as depicted in Fig. \ref{fig:ea_cv}.  This is precisely the setup in the well-known quantum superdense coding protocol \cite{Bennett-1992a}.  Letting $\cv^*(\mN)$ denote the entangled-assisted cv of $\mN$, another main result of ours is that $\cv^*(\mN)$ equals the conditional min-entropy \cite{Konig-2009a} of the Choi matrix of $\mN$.  Combining with the results of Duan and Winter \cite{Duan-2016a}, this implies that $\lceil\cv^*(\mN)\rceil$ captures the zero error classical communication cost for simulating $\mN$ when Alice and Bob have \textit{quantum} non-signaling resources.  Note that by additivity of the min-entropy, it follows that $\log\cv^*(\mN)$ is additive and therefore $\mc{CV}^*(\mN)=\log\cv^*(\mN)$.  A summary between channel cv and zero-error channel simulation is given in Table \ref{table:cv-and-chan-sim}.

\begin{center}
\begin{tabular}{|p{2cm}|p{6cm}|}\hline \label{table:cv-and-chan-sim}
  $\lceil\cv(\mN)\rceil$ & Classical communication cost  to perfectly simulate every classical channel induced by $\mN$ using classical non-signaling resource. \\ \hline
  $\lceil\cv^*(\mN)\rceil$ \Tstrut & Classical communication cost to perfectly simulate $\mN$ using quantum non-signaling resource.\\ \hline
\end{tabular}
\end{center}

This paper is structured as follows.  We begin in Section \ref{Sect:Notation} by introducing the notation used in this manuscript and reviewing some preliminary concepts.  Section \ref{Sect:cv-characterization} takes a deeper dive into the definition of channel communication value and relates it to the geometric measure of entanglement and other information-theoretic quantities such as the conditional min-entropy.  Section \ref{Sect:qubits} focuses on qubit channels and provides an analytic expression for cv in terms of the correlation matrix of the channel's Choi matrix.  The Werner-Holveo family of channels is introduced in Section \ref{Sect:Werner-Holevo} and the cv is computed.  The question of cv multiplicativity is taken up in Section \ref{Sect:Multiplicativity} with examples of both multiplicativity and non-multiplicativity being presented.  Notably, the cv capacity is shown to take a single-letter form for entanglement-breaking channels, Pauli qubit channels, and the general depolarizing channel.  Section \ref{sec:entanglement_assisted_cv} introduces the notion of entanglement-assisted communication value and relates it to the conditional min-entropy of the Choi matrix.  Different relations to the communication value are considered in Section \ref{sec:relaxations_of_cv}, with a particular focus on the PPT communication value and computable examples that it supports.  In Section \ref{Sect:Numerics} we describe a procedure for numerically estimating the cv of a given channel, and we provide a link to our developed software package, which performs this estimation.  Finally, Section \ref{Sect:Capacity-NS} provides a discussion of our results as they relate to channel capacity and zero error channel simulation.

\section{Notation and Preliminaries}

\label{Sect:Notation}

This paper considers exclusively finite-dimensional quantum systems represented by Hilbert spaces $\mc{H}^A, \mc{H}^B,\cdots$ etc.  The collection of positive operators acting on Hilbert space $\mc{H}^A$ will be denoted by $\Pos(A)$, which consists of all hermitian operators $\text{Herm}(A)$ acting on $A$ with a non-negative eigenvalue spectrum.  The subset of these operators having unit trace constitute the collection of density operators for system $A$, and we denote this set by $\mc{D}(A)$.  We write $\Vert\cdot\Vert_\infty$ and $\Vert\cdot\Vert_1$ to indicate the spectral and trace norms of elements in $\Pos(A)$, respectively.  For bipartite systems, an operator $\Omega\in \Pos(AB)$ is called separable if it can be expressed as a positive combination of product states, $\Omega=\sum_it_i\op{\alpha_i}{\alpha_i}\otimes\op{\beta_i}{\beta_i}$, with $t_i\geq 0$, $\op{\alpha_i}{\alpha_i}\in\mc{D}(A)$, and $\op{\beta_i}{\beta_i}\in\mc{D}(B)$, and we let $\sep(A:B)$ denote the set of all separable operators on systems $A$ and $B$.  Classical systems can be incorporated into this framework by demanding that the density matrix of every classical state be diagonal in a fixed basis.  In general, we will label a classical system by $X$ or $Y$.

Quantum channels provide the basic building blocks of any dynamical system.  Mathematically, they are represented by CPTP maps, and we denote the set of CPTP maps from system $A$ to $B$ by $\cptp(A\to B)$.  The set $\cptp(A\to B)$ is isomorphic to the subset of $\Pos(AB)$ consisting of operators whose reduced density operator on system $A$ is the identity.  Specifically, for every $\mN\in\cptp(A\to B)$ its Choi matrix is the associated operator $J_{\mN}\in\Pos(AB)$ given by
\[J_{\mN}=\id\otimes\mN(\phi^+_{d_A}),\]
where $\id$ is the identity map and $\phi^+_{d_A}=\op{\phi^+_{d_A}}{\phi^+_{d_A}}=\sum_{i,j=1}^{d_A}\op{ii}{jj}$.  Note that $\ket{\phi^+_{d_A}}$ is proportional to the normalized $d_A$-dimensional maximally entangled state, and we write the latter as $\ket{\Phi^+_{d_A}}:=\frac{1}{\sqrt{d_A}}\ket{\phi^+_{d_A}}$.  The fact that $\mN$ is completely positive assures that $J_{\mN}\geq 0$, and the trace-preserving condition means that $\tr_B J_{\mN}=\mbb{I}^A$, where $\mbb{I}$ is the identity operator.  On the other hand, if $\tr_A J_{\mN}=\mbb{I}^B$ then $\mN$ is a unital map, meaning that $\mN(\mbb{I}^A)=\mbb{I}^B$.  More generally, we say a map is sub-unital if $\mN(\mbb{I}^A)\leq \mbb{I}^B$.

An important subclass of channels are known as entanglement-breaking (EB).  These are characterized by the property that $\mN^{A\to B}\otimes \id^C(\rho^{AC})\in\sep(B:C)$ for all $\rho^{AC}\in\Pos(AC)$.  It is not difficult to see that $\mN\in\cptp(A\to B)$ is EB if and only if $J_{\mN}\in\sep(A:B)$.  For any subset $S\subset \text{Herm}(A)$, we let 
\[S^*=\{\omega\in\text{Herm}(A)\;|\;\langle \omega,\tau\rangle:=\tr[\omega\tau]\geq 0\;\;\forall \tau\in S\}\]
denote the dual cone of $S$.  As a final bit of notation, we write $\exp(x)$ and $\log(x)$ to mean $2^x$ and $\log_2x$, respectively.

\section{Characterizing the Communication Value}

\label{Sect:cv-characterization}
 
Let us begin with a few remarks regarding Definition \ref{Defn:cv}.  First, every choice of optimal POVM $\{\Pi_y\}_{y=1}^{n'}$ and set of signal states $\{\rho_x\}_{x=1}^n$ is characterized by a labeling function $f:[n']\to[n]$ such that $\max_{x\in[n]}\tr[\Pi_y\rho_x]=\tr[\Pi_y\rho_{f(x)}]$.  If the range of $f$ is strictly contained in $[n]$, then we can replace the set $\{\rho_x\}_{x=1}^n$ with a smaller set of signal states such that the map $f$ is surjective.  Similarly, if $f$ is not one-to-one, then we can coarse-grain the $n'$ POVM elements so that each outcome uniquely identifies a signal state.  Hence without changing the cv, we can assume $f:[m]\to[m]$ is a bijection with $m=\min\{n,n'\}$, and so
\begin{equation}
\cv^{n\to n'}(\mc{N})=\cv^{m\to m}(\mc{N})\qquad\text{for}\qquad m=\min\{n,n'\}.
\end{equation}

Another observation is that
\begin{equation}
\cv^{m\to m}(\mc{N})\leq \cv^{m'\to m'}(\mc{N})\qquad\text{for}\qquad m\leq m',
\end{equation}
which follows from the fact that we can always trivially split a POVM to increase outcomes $\Pi\to \frac{1}{2}\Pi+\frac{1}{2}\Pi$, and we can always increase the size of our input set $\{\rho_x\}_{x=1}^m$ by adding the same state $\rho_x$ multiple times.  Finally, we note that
\begin{equation}
\label{Eq:cv-carethedory}
\cv(\mc{N})=\cv^{d_B^2\to d_B^2}(\mc{N}),
\end{equation}
where $d_B$ is the dimension of the output system.  This follows from the fact that any POVM on a $d_B$-dimensional systems can always be decomposed into a convex combination of extremal POVMs, each with at most $d_B^2$ outcomes \cite{Davies-1978a}, and the cv can always be attained with one of these extremal measurements.  

It is also not difficult to see that 
\begin{align}
\label{Eq:cv-bounds}
    1\leq \cv(\mN)\leq \min\{d_A,d_B\}
\end{align}
for any channel $\mN$.  The lower bound holds by considering a constant input $\rho_x=\rho$ (for all $x$) so that $\sum_x\tr[\Pi_x\mN(\rho_x)]=\sum_x\tr[\Pi_x\mN(\rho)]=1$, since $\mN$ is trace preserving and $\sum_x\Pi_x=\mbb{I}_{d_B}$.  Similarly, the upper bound follows from the inequalities
\begin{align}
    \sum_x\tr[\Pi_x\mN(\rho_x)]&\leq \sum_x\tr[\Pi_x]=\tr[\mbb{I}_{d_B}]=d_B;\\
    \sum_x\tr[\Pi_x\mN(\rho_x)]&= \sum_x\tr[\mN^\dagger(\Pi_x)\rho_x]\notag\\
    &\leq \sum_x\tr[\mN^\dagger(\Pi_x)]=\tr[\mbb{I}_{d_A}]=d_A,
\end{align}
where $\mN^\dagger:B\to A$ is the adjoint map of $\mN$, and so $\{\mN^\dagger(\Pi_x)]\}_x$ will always be a valid POVM on Alice's system.

Notice that when $\cv(\mN)=1$, Bob can do no better than randomly guessing Alice's input.  The following proposition characterizes the type of channel for which this is the case.
\begin{proposition}
$\cv(\mN)=1$ iff $\mN$ is a replacer channel; i.e. there exists a fixed state $\sigma$ such that $\mN(\rho)=\sigma$ for all states $\rho$.  
\end{proposition}
\begin{proof}
If $\mN$ is a replacer channel, then clearly $\cv(\mN)=1$.  On the other hand, suppose that $\mN$ is not a replacer channel.  This means there exists two inputs $\rho_1$ and $\rho_2$ such that $\Delta=\mN(\rho_1)-\mN(\rho_2)\not=0$.  Hence by performing a Helstrom measurement on the channel output (i.e. projecting onto the $\pm$ parts of $\Delta$) \cite{Helstrom-1976a}, one obtains $\cv(\mN)\geq 1+\tfrac{1}{2}\Vert\Delta\Vert_1>1$. 
\end{proof}

\subsection{Communication Value via Conic Optimization}

In general $\cv(\mc{N})$ is difficult to compute.  This can be seen more explicitly by casting $\cv(\mc{N})$ as an optimization over the separable cone, $\sep(A:B)$, whose membership is NP-Hard to decide \cite{Gurvits-2003a}.  
Nevertheless, expressing $\cv(\mc{N})$ as an optimization over $\sep(A:B)$ leads to computable upper bounds since there are well-known relaxations to the set $\sep(A:B)$ that are easier to handle analytically.
\begin{proposition}
\label{Prop:Proposition-cv}
For $\mc{N}\in\cptp(A\to B)$, 
\begin{align}
\cv(\mc{N})=\max &\;\;\tr[\Omega^{AB}J_\mc{N}]\notag\\
\text{subject to}\;&\;\;\tr_A[\Omega^{AB}]=\mbb{I}^B;\notag\\
&\;\;\Omega^{AB}\in\sep(A:B).
\label{Eq:Proposition-cv}
\end{align}
\end{proposition}
\begin{proof}
By Eq. \eqref{Eq:cv-carethedory} we have
\begin{align}
\cv(\mc{N})&=\max_{\{\Pi_x\}, \{\rho_x\}}\sum_{x=1}^{d_B^2}\tr[\Pi_x\mc{N}(\rho_x)]\notag\\
&=\max_{\{\Pi_x\}, \{\rho_x\}}\sum_{x=1}^{d_B^2}\tr[\rho_x^T\otimes\Pi_x(J_{\mc{N}})]\notag\\
&=\max \tr[\Omega^{AB}J_{\mc{N}}],
\end{align}
where $\Omega^{AB}=\sum_{x=1}^{d_B^2}\rho_x^T\otimes\Pi_x$ satisfies the conditions of Eq. \eqref{Eq:Proposition-cv}.  Conversely, any $\Omega^{AB}\in\sep(A:B)$ can be written as $\Omega^{AB}=\sum_{x}\op{\psi_x}{\psi_x}^A\otimes\omega_x^B$ with $\ket{\psi_x}$ being a pure state.  The condition $\tr_A[\Omega^{AB}]=\mbb{I}^B$ implies that $\{\omega_x\}_x$ constitutes a POVM.
\end{proof}
\noindent Note that strong duality holds for the conic program here, and so Proposition \ref{Prop:Proposition-cv} can be cast in dual form as
\begin{align}
\cv(\mc{N})=\min &\;\;\tr[Z^{B}]\notag\\
\text{subject to}&\;\;\mbb{I}^{A}\otimes Z^{B}-J_{\mc{N}}^{AB}\in\sep^*(A:B).
\label{Eq:cv-dual}
\end{align}


In Section \ref{Sect:Capacity-NS}, we will explore different relaxations to this problem by considering outer approximations of $\sep(A:B)$.  For example, the cone $\ppt(A:B)$, which consists of all bipartite positive operators having a positive partial transpose, contains $\sep(A:B)$ \cite{Peres-1996a}.  Replacing $\sep(A:B)$ with $\ppt(A:B)$ in Eq. \eqref{Eq:Proposition-cv} gives us a semi-definite program (SDP).  Furthermore, since $\sep(A:B)=\ppt(A:B)$ whenever $d_Ad_B\leq 6$ \cite{Horodecki-1996a}, we thus obtain the following.
\begin{corollary}
\label{cor:Proposition-PPT}
Suppose $\mc{N}\in\cptp(A\to B)$ with $d_Ad_B\leq 6$.  Then
\begin{align}
\cv(\mc{N})=\max &\;\;\tr[\Omega^{AB}J_\mc{N}]\notag\\
\text{subject to}\;&\;\;\tr_A[\Omega^{AB}]=\mbb{I}^B\notag\\
&\;\;\Omega^{AB}\in\ppt(A:B).
\label{Eq:Proposition-PPT}
\end{align}
\end{corollary}

\subsection{An Entropic Characterization of Communication Value}
\label{sec:entropic-characterization-of-cv}

\subsubsection{The Conditional Separable Min-Entropy}

An alternative but related manner of characterizing the communication value is in terms of the min-entropy or variations of it. Equation \eqref{Eq:cv-dual} might strike the reader as closely resembling the conditional min-entropy of $J_\mc{N}$.  Recall that the conditional min-entropy of a positive bipartite operator $\omega^{AB}$ is given by
\begin{align}
    H_{\min}(A|B)_{\omega}=-\min_{\sigma^B\in\mc{D}(B)}D_{\max}(\omega\Vert \mbb{I}^A\otimes \sigma^B),
\end{align}
where $D_{\max}(\mu\Vert\nu)=\min\{\lambda\;|\;\mu\leq 2^\lambda \nu\}$ \cite{Konig-2009a}.  Here $\leq$ denotes a generalized inequality over the convex cone of positive-semidefinite operators; i.e. $X\leq Y$ iff $Y-X\in \Pos(AB)$.  Equivalently, we can combine the two minimizations in the definition of $H_{\min}$ to write
\begin{align}
    \exp[-H_{\min}(A|B)_{\omega}]=\min &\;\;\tr[Z^{B}]\notag\\
\text{subject to}&\;\;\mbb{I}^{A}\otimes Z^{B}-\omega^{AB}\in\Pos(AB).
\end{align}
Comparing with Eq. \eqref{Eq:cv-dual}, we see that cv is recovered by changing the cone from $\Pos(AB)$ to $\sep^*(A:B)$.  Let us denote the cone inequality over $\sep^*(A:B)$ by $\leq_{\sep^*}$ such that $X\leq_{\sep^*}\! Y$ iff $Y-X\in\sep^*(A:B)$.  Then we can introduce a restricted conditional min-entropy.
\begin{definition}\label{defn:sep-min-entropy}
The conditional \textit{separable} min-entropy of a positive bipartite operator $\omega^{AB}$ is defined as
\begin{equation}\label{eqn:sep-min-ent-def}
   H_{\min}^{\sepl}(A|B)_{\omega}=-\min_{\sigma^B\in\mc{D}(B)}D^{\sepl}_{\max}(\omega\Vert \mbb{I}^A\otimes \sigma^B),
\end{equation}
where $D^{\sepl}_{\max}(\mu\Vert\nu)=\min\{\lambda\;|\;\mu\leq_{\sep^*}\! 2^\lambda \nu\}$.
\end{definition}
\noindent By Eq. \eqref{Eq:cv-dual}, we therefore have
\begin{equation} \label{eqn:cv-sep-min-relation}
   \cv(\mN)=\exp[-H^{\sepl}_{\min}(A|B)_{J_\mc{N}}].
\end{equation}

The separable min-entropy enjoys a data-processing inequality under one-way LOCC from Bob to Alice.  The latter consists of any bipartite map $\Phi\in \cptp(AB \to A'B')$ having the form $\Phi=\sum_i\mN_i\otimes\mM_i$, where $\mN_i\in\cptp(A\to A')$ and $\sum_i\mM_i\in \cptp(B\to B')$ with each individual $\mM_i$ being CP.  In fact we can prove the data-processing inequality under an even larger class of operations.
\begin{proposition}\label{eqn:sep-min-DPI}
Let $\Phi:\Pos(AB)\to\Pos(A'B')$ be any positive map whose adjoint is non-entangling (i.e. $\Phi^\dagger:\sep(A':B')\to\sep(A:B)$) and that further satisfies $\Phi(\mbb{I}^A\otimes \sigma^{B'})\leq\mbb{I}^{A'}\otimes\phi(\sigma^{B'})$ for some trace-preserving map $\phi:\Pos(B)\to\Pos(B')$.  Then $H^{\sepl}_{\min}(A|B)_{P} \leq H^{\sepl}_{\min}(A|B)_{\Phi(P)} $ for all $P \in \mathrm{Pos}(AB)$.
\end{proposition}
\begin{proof}
Since $\Phi^\dagger$ preserves separability, we must have $\Phi(Q)\in\sep^*(A':B')$ for all $Q\in \sep^*(A:B)$.  Hence, 
\begin{align}
   2^{\lambda}\mbb{I} \otimes \sigma-P\geq_{\sep^*} 0\;&\Rightarrow\;2^{\lambda}\Phi(\mbb{I} \otimes \sigma)-\Phi(P)\geq_{\sep^*} 0\notag\\
   &\Rightarrow\;2^{\lambda}\mbb{I} \otimes \phi(\sigma)-\Phi(P)\geq_{\sep^*} 0.\notag
\end{align}
In other words, any feasible pair $(\sigma,\lambda)$ in the minimization of $H_{\min}^{\sepl}(A|B)_P$ also leads to a feasible pair for $H_{\min}^{\sepl}(A|B)_{\Phi(P)}$.  The $-1$ factor in the definition of $H_{\min}^{\sepl}$ then implies the proposition.
\end{proof}
\noindent The maps of Proposition \ref{eqn:sep-min-DPI} include those of the form $\Phi=\mM\otimes\mN$, where $\mM$ is sub-unital and $\mN$ is CPTP.  These maps are known to satisfy the data-processing inequality for the standard min-entropy \cite{Tomamichel-2015}.  However we suspect that Proposition \ref{eqn:sep-min-DPI} includes maps for which the standard min-entropy data-processing inequality does \textit{not} hold.

We can apply Proposition \ref{eqn:sep-min-DPI} to the processing of Choi matrices.  However, in this case not all maps $\Phi$ satisfying the conditions of Proposition \ref{eqn:sep-min-DPI} are physically meaningful.  Specifically, we require the additional condition that $\tr_{B'}\Phi(P)=\mbb{I}^{A'}$ for all operators $P$ in which $\tr_{B}P=\mbb{I}^A$.  This assures that $\Phi$ maps Choi matrices to Choi matrices.  One particular class of maps having this form are those in which $\Phi$ is a product of a positive unital map and a CPTP map, i.e. $\Phi=\mc{E}_{\text{pre}}^\dagger \otimes\mc{E}_{\text{post}}$.  In this case, $\mc{E}_{\pre}$ and $\mc{E}_{\post}$ are pre- and post-processing maps for a given channel, respectively.  As a consequence of Proposition \ref{eqn:sep-min-DPI} we therefore observe the following corollary, which can also be seen directly from the definition of communication value.
\begin{corollary}\label{corr:pre-post-proc}
Communication value is non-increasing under pre- and post-processing of the channel.
\end{corollary}

Note that for classical systems $X$ and $Y$, we have $\Pos(XB)=\sep(X:B)$ and $\Pos(AY)=\sep(A:Y)$.  These correspond to classical-to-quantum and quantum-to-classical channels, respectively, and in these cases Eq. \eqref{eqn:cv-sep-min-relation} reduces to
\begin{align}
    \cv(\mN^{X\to B})&=\exp[-H_{\min}(X|B)_{J_\mN}]\label{Eq:cv-cq}\\
    \cv(\mN^{A\to Y})&=\exp[-H_{\min}(A|Y)_{J_\mN}]\label{Eq:cv-qc}.
\end{align}

\subsubsection{The max-Holevo Information}

The cv can be further related to the max-Holevo information of a channel, $\chi_{\max}(\mN)$.  This quantity has been introduced in the study of ``sandwiched'' R\'{e}nyi divergences \cite{Wilde-2014a, Beigi-2013a} and is defined as
\begin{align*}
\chi_{\max}(\mN) = \max_{\rho^{XA}} \min_{\sigma^{B}} D_{\max}(\rho^{XB}||\rho^{X} \otimes \sigma^{B}),
\end{align*}
where the maximization is taken over all cq states $\rho^{XA}=\sum_xp(x)\op{x}{x}\otimes\rho_x^A$ and
\[\rho^{XB}:=\sum_xp(x)\op{x}{x}\otimes\mN(\rho_x^A).\]
In fact, since $D_{\max}$ is quasi-convex (i.e. $D_{\max}(\sum_ip(i)\rho_i\Vert\sum_ip(i)\sigma_i)\leq\max_i D_{\max}(\rho_i\Vert\sigma_i)$ \cite{Datta-2009a}), if follows that we can restrict attention to pure $\rho_x^A=\op{\psi_x}{\psi_x}^A$ in the definition of $\chi_{\max}$.  Letting $U$ be the unitary such that $U\ket{x}=\ket{\psi_x}$, the maximization over $\rho^{XA}$ can then be replaced by a maximization over $U$ such that
\begin{equation}
\label{Eq:cq-state-simplify}
    \rho^{XA}=(\mbb{I}\otimes U)\sum_{x}p(x)\op{xx}{xx}(\mbb{I}\otimes U)^\dagger.
\end{equation}

We use this simplification to prove a relationship between channel $\chi_{\max}$ and conditional $H^{\sepl}_{\min}$.
\begin{theorem}
\label{Thm:chi-max-hmin}
For any channel $\mN^{A\to B}$,
\begin{equation}
    \chi_{\max}(\mN)=\log\cv(\mN).
\end{equation}
\end{theorem}
\begin{proof}
Using Eq. \eqref{Eq:cq-state-simplify} we have    \[\rho^{XB}=\sum_xp(x)\op{x}{x}\otimes\rho_x,\] where $\rho_x=\mN(U\op{x}{x}U^\dagger)$.  Since $\rho^X=\sum_xp(x)\op{x}{x}$, the definition of $D_{\max}$ yields
\begin{align}
   D_{\max}(\rho^{XB}\Vert\rho^X\otimes\sigma^B)&=\min\{\lambda\;|\;\rho_x\leq 2^\lambda \sigma,\;\forall x\}\notag\\
   &=D_{\max}(\wt{\rho}^{XB}\Vert\mbb{I}\otimes\sigma^B)\notag\\
   &=D_{\max}^{\sepl}(\wt{\rho}^{XB}\Vert\mbb{I}\otimes\sigma^B),\label{Eq:Dmax-ineq}
\end{align}
where
\begin{align}
    \wt{\rho}^{XB}&=\sum_x\op{x}{x}\otimes\mN(U\op{x}{x}U^\dagger) \notag\\
    &=\Delta_{U^T}\otimes\id^B(J_\mN),
\end{align}
and $\Delta_{U^T}(\tau)=\sum_{x}\op{x}{x}U^T(\tau) U^*\op{x}{x}$ is a completely dephasing map after applying the rotation $U^T$.  The last equality in Eq. \eqref{Eq:Dmax-ineq} follows from the fact that $\Pos(XB)=\sep(X:B)$, as noted above.  Then by data-processing (Proposition \ref{eqn:sep-min-DPI}), we have
\begin{align}
    D_{\max}^{\sepl}(\Delta_{U^T}\otimes\id^B(J_\mN)\Vert \mbb{I}\otimes\sigma^B)\leq D_{\max}^{\sepl}(J_\mN\Vert \mbb{I}\otimes\sigma^B) \notag
\end{align}
for any $\sigma^B$ and unitary $U$ on system $A$.  Hence from the definitions it follows that
\begin{equation}
    \chi_{\max}(\mN)\leq -H_{\min}^{\sepl}(A|B)_{J_{\mN}}=\log\cv(\mN).
\end{equation}

To prove the reverse inequality, for arbitrary $\sigma^B$ let $\lambda_0=D_{\max}^{\sepl}(J_\mN\Vert\mbb{I}\otimes\sigma^B)$.  Hence
\begin{equation}
    2^{\lambda_0}\mbb{I}\otimes\sigma^B-J_{\mc{N}}\in\sep^*(A:B),
\end{equation}
and since $\lambda_0$ is a minimizer, there must exist some product state $\ket{\alpha}\ket{\beta}$ such that
\begin{equation}
    2^{\lambda_0}\bra{\beta}\sigma^B\ket{\beta}=\bra{\beta}\mN(\op{\alpha^*}{\alpha^*})\ket{\beta}.
\end{equation}
Let $U$ be any unitary that rotates $\{\ket{x}\}_{x=1}^{d_A}$ such that $U\ket{1}=\ket{\alpha^*}$.  Therefore
\begin{align}
   2^{\lambda_0}\bra{\alpha,\beta}\mbb{I}\otimes\sigma^B\ket{\alpha,\beta}= \bra{\alpha,\beta}\Delta_{U^T}\otimes\id^B(J_\mN)\ket{\alpha,\beta}, 
\end{align}
which means that $D_{\max}^{\sepl}(\Delta_{U^T}\otimes\id^B(J_\mN)\Vert \mbb{I}\otimes\sigma^B)$ can be no less than $\lambda_0=D_{\max}^{\sepl}(J_\mN\Vert\mbb{I}\otimes\sigma^B)$.  Since this holds for all $\sigma^B$ and we are maximizing over $U$, we have
\begin{equation}
\chi_{\max}(\mN)\geq -H_{\min}^{\sepl}(A|B)_{J_{\mN}}=\log\cv(\mN).
\end{equation}
\end{proof}

We close this section by providing an alternative proof of Theorem \ref{Thm:chi-max-hmin}.  Instead of going through the Choi matrix, the following argument relies on a characterization of cv in terms of maximizing the min-entropy over encodings. In some sense this is intuitive as the communication value is optimizing minimal error discrimination, which min-entropy characterizes \cite{Konig-2009a}.  For this reason, the conceptual underpinning of this alternative derivation may be of interest in other applications.

 Let $\{\rho_{x}^{A}\}$ denote a subset of states for some alphabet $\mX$, $\rho^{XA}$ be a cq state defined using $\{\rho_{x}^{A}\}$, $\rho_{U}$ be the maximally mixed state on the relevant space, and $\rho^{XB} := (\mathrm{id}^{X} \otimes \mN)(\rho^{XA})$. Starting from \eqref{Eq:cv-dual},
 \begin{align}
    \cv(\mN) 
    =& \min\{\tr[Z^{B}] : \mbb{I}^{A} \otimes Z^{B} \geq_{\sep^{*}} J_{\mN} \} \notag \\
    =& \sup_{\{\rho_{x}^{A}\}} \min\{\tr[Z^{B}] : Z^{B} \geq \mN(\rho_{x}^{A}) \, \forall x \in \mX \} \notag \\
    =& \sup_{\substack{\rho^{XA}: \\ \rho^{X} = \rho_{U}}} |\mX| \min\{\tr[\wt{Z}^{B}] : \mbb{I}^{A} \otimes \wt{Z}^{B} \geq \rho^{XB} \} \notag \\
    =& \sup_{\substack{\rho^{XA}: \\ \rho^{X} = \rho_{U}}} |\mX| \exp(-H_{\min}(X|B)_{\rho^{XB}}) \notag \\
    =& \sup_{\{\rho_{x}^{A}\}} \min_{\sigma^B} \lambda_{\max}(\sum_{x} \op{x}{x} \otimes \sigma^{-1/2 } \rho^{B}_{x} \sigma^{-1/2}) \notag \\ 
    =& \sup_{\rho^{XA}} \min_{\sigma^{B}} \exp( D_{\max}(\rho^{XB}||\rho^{X} \otimes \sigma^{B})) \notag \\
    =& \exp(\chi_{\max}) \label{eqn:cv-max-holevo-alt} \ ,
 \end{align}
 where the second equality is using $X^{AB} \in \sep^{*} \Leftrightarrow \bra{\alpha}\bra{\beta}X\ket{\alpha}\ket{\beta}$ for all unit vectors $\alpha,\beta$ and the action of the channel in terms of the Choi, the third is by using uniform probability on $\mX$, the fourth is by definition of min-entropy, the fifth is using $D_{\max}(\rho||\sigma) = \lambda_{\max}(\sigma^{-1/2}\rho \sigma^{-1/2})$, the sixth is using that $p(x)^{-1/2}p(x)p(x)^{-1/2} = 1$, and the final equality is by definition.

\subsection{Communication Value in Terms of Singlet Fraction}

Another advantage of viewing $\cv(\mN)$ in terms of a restricted min-entropy is that it provides an alternative operational interpretation of the communication value in terms of the singlet fraction, which it inherits from the min-entropy conic program.  Recall that for a bipartite density matrix $\omega^{AB}$, its $d_A$-dimensional singlet fraction is defined as
\begin{equation}
    F^+_{d_A}(\omega)=\max_U\bra{\Phi^+_{d_A}}(\mbb{I}^A\otimes U^B)\omega^{AB}(\mbb{I}^A\otimes U^B)^\dagger\ket{\Phi^+_{d_A}},\notag
\end{equation}
where the maximization is taken over all unitaries applied to system $B$ \cite{Bennett-1996a}.
\begin{proposition}
$\cv(\mN)$ is the maximum singlet fraction achievable using an entanglement-breaking channel after the action of $\mN$ on $\ket{\Phi^+_{d_A}}$.
\end{proposition}
\begin{proof}
This follows the proof of the operational interpretation of the min-entropy \cite{Konig-2009a}, and we walk through the argument again here to exemplify that the only change is in restricting to the separable cone. Proposition \ref{Prop:Proposition-cv} shows that $\cv(\mc{N})$ is the maximum value $\tr[\Omega^{AB}J_{\mc{N}}]$, where $\Omega^{AB}$ is the Choi matrix of a unital (entanglement-breaking) map, i.e. $\Omega^{AB} = J_{\mM}$ for some entanglement-breaking unital map $\mM$. Thus,
\begin{align*}
\cv(\mN)&= \langle J_{\mM}, J_{\mN} \rangle \\
=& d_{A}^{2}\langle (\id \otimes \mM)(\widehat{\Phi}^{+}), (\id \otimes \mN)(\widehat{\Phi}^{+}) \rangle \\ 
=& d_{A}^{2} \langle \widehat{\Phi}^{+}, (\id \otimes \mM^\ast \circ \mN)(\Phi^{+}) \rangle \ ,
\end{align*}
where we have used the definitions of the Choi matrix and adjoint map, and $\widehat{\Phi}^{+}$ is the normalized maximally entangled state. Noting the adjoint of an entanglement-breaking map is entanglement-breaking, and the adjoint of a unital map is trace-preserving, we have
\begin{align}
\cv(\mN) &=d_{A}^2\max_{\mc{E}\in\eb(B\to A)}F^+_{d_A}\left(\id \otimes \mc{E} \circ \mN(\Phi_{d_A}^+)\right) \notag\\
&=d_{A}^2\max_{\mc{E}\in\eb(B\to A)}\bra{\Phi^+_{d_A}}\id \otimes \mc{E} \circ \mN(\Phi_{d_A}^+)\ket{\Phi^+_{d_A}},\label{Eq:eb-recoverability}
\end{align}
where the last line follows from the fact that the maximization over local unitaries in the definition of $F^+_{d_A}$ can be included in the maximization over entanglement-breaking channels.
\end{proof}
\begin{figure}\label{fig:CVMeasure}
    \begin{minipage}[t]{0.5\linewidth}
    \end{minipage}
    \begin{center}
        \begin{tikzpicture}
            \tikzstyle{porte} = [draw=black!50, fill=black!20]
                \draw
                (0,0) node (origin) {$\ket{\Phi^{+}}$}
                ++(2, 0.5) node[porte] (m1) {$\mathcal{N}$}
                ++(2, 0) node[porte] (m2) {$\Psi$}
                ++(-1,-1) node[porte] (m3) {$\id$}
                ++(0,-0.75) node (m4) {Versus}
                ++(-3,-1.25) node(m5) {$\ket{\Phi^{+}}$}
                ++(3, 0.5) node[porte] (m6) {$\id$}
                ++(0, -1) node[porte] (m7) {$\id$}
                ;
                \draw
                (m3) -- (1,-0.5) -- (0.5,0) -- (1,0.5) -- (m1)
                (m1) -- (m2)
                (m2) edge[-] ++(0.75,0)
                (m3) edge[-] ++(1.75,0)
                (m7) -- (1,-3) -- (0.5,-2.5) -- (1,-2) -- (m6)
                (m6) edge[-] ++(1.5,0)
                (m7) edge[-] ++(1.5,0)
                ;
        \end{tikzpicture}
    \end{center}
    \caption{The channel value is a measure of the maximum singlet fraction achievable using an entanglement-breaking (EB) channel to recover the singlet after the action of $\mN$. In this setting, non-multiplicativity means that the optimal EB channel $\Psi$ changes when using $\mN$ in parallel such that the achievable singlet fraction increases. This suggests that $\cv(\mN)$ is a measurement of entanglement preservation.}
\end{figure}
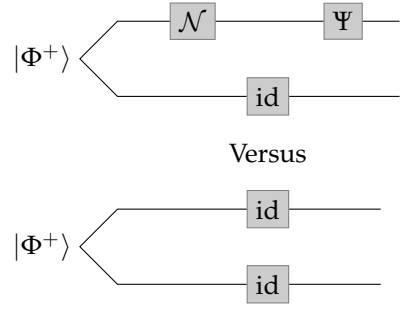
Equation \eqref{Eq:eb-recoverability} yields the interpretation of $\cv(\mN)$ as how well a maximally entangled state can be recovered after Alice sends one half of $\ket{\Phi^+_{d_A}}$ to Bob over the channel $\mN$, and he is limited to performing an entanglement-breaking channel as post-processing error correction (see Fig. \ref{fig:CVMeasure}).  In Section \ref{sec:entanglement_assisted_cv}, it will be shown that the entanglement-assisted communication value is characterized by $\exp(-H_{\min}(A|B)_{J_{\mN}})$, and it thus has a similar operational interpretation except Bob is now able to perform an arbitrary quantum channel to try and recover the maximally entangled state.  Moreover, in Section \ref{sec:relaxations_of_cv}, when we consider relaxations of cv to other cones, we will find that the PPT min-entropy $H_{\min}^{\ppt}$ also retains this operational interpretation, but with recovery being relaxed to the use of co-positive maps. 

\subsection{The Geometric Measure of Entanglement and Maximum Output Purity}

The channel cv is closely related to the geometric measure of entanglement (GME) \cite{Shimony-1995a, Wei-2003a}.  For a bipartite positive operator $\omega^{AB}$, its GME is defined as $G(\omega)=-\log\Lambda^2(\omega)$, where
\begin{align}
\Lambda^2(\omega)=\max_{\ket{\alpha}\ket{\beta}}\bra{\alpha,\beta}\omega\ket{\alpha,\beta}.
\end{align}
We can phrase this as the SDP
\begin{align}
\Lambda^2(\omega)=\max &\;\;\tr[\sigma^{AB}\omega]\notag\\
\text{subject to}&\;\;\tr[\sigma^{AB}]=1\notag\\
&\;\;\sigma^{AB}\in\sep(A:B).
\label{Eq:GME-SDP}
\end{align}
For a channel $\mc{N}$ with Choi matrix $J_\mc{N}$, this SDP can be expressed in dual form as
\begin{align}
\Lambda^2(J_{\mc{N}})=\min &\;\;\lambda\notag\\
\text{subject to}&\;\;\lambda\mbb{I}^{AB}-J_{\mc{N}}^{AB}\in\sep^*(A:B).
\label{Eq:GM-dual}
\end{align}
For any dual feasible $\lambda$ in Eq. \eqref{Eq:GM-dual}, we can take $Z=\lambda\mbb{I}^B$ in Eq. \eqref{Eq:cv-dual} to obtain
\begin{align}
\cv(\mc{N})\leq d_B\Lambda^2(J_\mc{N}).
\end{align}
On the other hand, suppose that $Z$ is dual feasible in Eq. \eqref{Eq:cv-dual}.  Then since $\mbb{I}^A\otimes (\lambda\mbb{I}^B-Z)\in\sep^*(A:B)$ for $\lambda=\Vert Z\Vert_\infty$, we have that
\begin{align}
\lambda\mbb{I}^A\otimes \mbb{I}^B-J_{\mc{N}}^{AB}&=\mbb{I}^A\otimes (\lambda\mbb{I}-Z)^B+\mbb{I}^A\otimes Z^B-J_{\mc{N}}^{AB}\notag\\
&\in\sep^*(A:B).
\end{align}
Hence $\lambda \mbb{I}^{AB}$ is dual feasible in Eq. \eqref{Eq:GM-dual}.  Since $\lambda\leq \tr[Z]$ for every $Z$, we conclude
\begin{align}
\label{Eq:cv-GM}
\Lambda^2(J_{\mc{N}})\leq\cv(\mc{N})\leq d_B\Lambda^2(J_\mc{N}).
\end{align}
While these relationships are somewhat obvious from the formulation of the problem, it is nice to see them explicitly falling out of the two conic programs.

It is known that $\Lambda^2(J_\mN)$ is equal to the maximum output purity of the channel $\mc{N}$ \cite{Werner-2002a, Zhu-2011a} which is defined as 
\[\nu_{\infty}(\mN) := \underset{\rho \in \mc{D}(A)}{\sup} \|\mN(\rho)\|_{\infty}.\]
To see the equivalence, first note that this supremum is attained for a pure-state input due to convexity of the operator norm.  Then 
\begin{align}
\nu_{\infty}(\mN)&= \underset{\ket{\alpha}}{\sup} \|\mN(\op{\alpha}{\alpha})\|_{\infty}\notag\\
     &= \underset{\ket{\alpha}}{\sup}\sup_{\ket{\beta}}\bra{\beta}\mN(\op{\alpha}{\alpha})\ket{\beta}\notag\\
     &=\sup_{\ket{\alpha},\ket{\beta}}\bra{\alpha^*,\beta}J_{\mN}\ket{\alpha^*,\beta}\notag\\
     &=\Lambda^2(J_{\mN}). \label{eqn:GME-max-output-purity-equiv}
\end{align}     
An alternative way to prove this equality is using \cite{Werner-2002a}. In that work they establish that $\nu_{\infty}(\mN) = \Lambda^{2}(\ket{\mN})$ where $\ket{\mN}$ is an un-normalized vector induced by the Kraus representation of $\mN$. One can in fact show that $\ket{\mN}$ is the purification of the Choi matrix, though this has not been stated previously to the best of our knowledge. As the GME of a pure state is the same as the GME of the pure state with a single register traced off \cite{Jung-2008}, one can conclude $\nu_{\infty}(\mN) = \Lambda^{2}(J_{\mN})$. Despite $\Lambda^2(J_\mN)=\nu_\infty(\mN)$ being a lower bound of $\cv(\mN)$ in Eq. \eqref{Eq:cv-GM}, in general this bound will not be tight.  Hence communication value is capturing property of a quantum channel that is distinct from maximum output purity.  In fact, we have the following.
\begin{proposition}
$\nu_\infty(\mN)=\cv(\mc{N})$ iff $\mc{N}$ is a replacer channel.
\end{proposition}
\begin{proof}
If $\mc{N}$ is a replacer channel, say, $\mc{N}(\rho)=\op{\beta}{\beta}$ $\forall \rho$, then clearly $\nu_\infty(\mN)=\cv(\mc{N})=1$.  On the other hand, suppose that $\nu_\infty(\mN)=\cv(\mc{N})$ and let $\ket{\alpha}\ket{\beta} := \mathrm{argmax}(\Lambda^{2}(J_{\mN}))$.  Then for an arbitrary state $\rho^A\in\mc{D}(A)$ consider the operator
\begin{align}\notag
    \Omega^{AB}=\op{\alpha}{\alpha}\otimes\op{\beta}{\beta}+\rho\otimes(\mbb{I}-\op{\beta}{\beta}).
\end{align}
Note that since $\tr_A\Omega^{AB}=\mbb{I}^B$ and $\Omega^{AB}\in\sep(A:B)$, it is a feasible solution for the optimization of $\cv(\mN)$ in Proposition \ref{Prop:Proposition-cv}.  Hence for all $\rho$ we have
\begin{align}
    \cv(\mc{N})&\geq \tr[\Omega^{AB} J_{\mc{N}}]\notag\\
    &=\bra{\alpha,\beta}J_{\mN}\ket{\alpha,\beta}+\tr[\mN(\rho^*)(\mbb{I}-\op{\beta}{\beta})]\notag\\&=\nu_\infty(\mN)+\tr[\mN(\rho^*)(\mbb{I}-\op{\beta}{\beta})].
\end{align}
But by the assumption $\nu_\infty(\mN)=\cv(\mc{N})$, the second term must vanish for all $\rho$.  This means that $\mN$ is a replacer channel, outputting $\op{\beta}{\beta}$ for all its trace-one inputs.
\end{proof}

\section{Examples}

Having characterized the channel cv in a variety of different ways, we now focus on the problem of computing it.  In general, this is a challenging task.  Here we provide closed-form solutions for arbitrary qubit channels and the family of Holevo-Werner channels.  The latter will also provide a useful case study when we study relaxations to the communication value in Section \ref{sec:relaxations_of_cv}.

\subsection{Qubit Channels}

\label{Sect:qubits}

Every qubit channel $\mc{N}$ induces an affine transformation on the Bloch vector of the input state.  In more detail, every positive operator $\rho$ can be written as $\rho=\gamma(\mbb{I}+\mbf{r}\cdot\hat{\sigma})$, where $\mbf{r}\in\mbb{R}^3$ has norm no greater than one.  Then when $\mc{N}$ acts on $\rho$, it induces an affine transformation $\mbf{r}\mapsto A\mbf{r}+\mbf{c}$, with $A$ being some $3\times 3$ matrix and $\mbf{c}\in\mbb{R}^3$.  
Now, let $\sigma=\sum_{k}\alpha^T_k\otimes \beta_k$ be an arbitrary two-qubit separable operator with $\tr[\alpha_k]=1$ and $\sum_k\beta_k=\mbb{I}$.  We given them Bloch sphere representations $\alpha_k=\frac{1}{2}(\mbb{I}+\mbf{a}_k\cdot\vec{\sigma})$ and $\beta_k=\gamma_k(\mbb{I}+\mbf{b}_k\cdot\vec{\sigma})$ so that
\begin{align}
\tr[\sigma J_\mc{N}]&=\sum_k\tr[\beta_k\mc{N}(\alpha_k)]\notag\\
&=\sum_k\frac{\gamma_k}{2}\tr[(\mbb{I}+\mbf{b}_k\cdot\vec{\sigma})(\mbb{I}+(A\mbf{a}_k+\mbf{c})\cdot\vec{\sigma})]\notag\\
&=1+\sum_k\gamma_k\mbf{b}_k\cdot (A\mbf{a}_k+\mbf{c})\notag\\
&=1+\sum_k\gamma_k\mbf{b}_k^TA\mbf{a}_k,
\end{align}
where the last equality follows from the fact that $\sum_k\gamma_k\mbf{b}_k=0$ since $\sum_k\beta_k=\mbb{I}$.  Our task then is to maximize $\sum_k\gamma_k\mbf{b}_k^TA\mbf{a}_k$ under the constraints that (i) $\sum_k\gamma_k\mbf{b}_k=0$, (ii) $\sum_k\gamma_k=1$, and (iii) $\Vert\mbf{b}_k\Vert,\Vert\mbf{a}_k\Vert\leq 1$.  It is easy to see that this maximization is attained by taking $\mbf{b}_1$ and $\mbf{b}_2$ as anti-parallel unit vectors aligned with the left singular vectors of $A$ corresponding to its largest singular value, and likewise for $\mbf{a}_1$ and $\mbf{a}_2$ with respect to the right singular vector.  Additionally, taking $\gamma_k=\frac{1}{2}$ for $k=1,2$ satisfies all the conditions.  Hence we have the following.
\begin{theorem}
\label{Thm:qubit-cv}
For a qubit channel $\mc{N}$, let $A$ be the $3\times 3$ correlation matrix of $J_\mc{N}$; i.e. $A_{ij}=\frac{1}{2}\tr[(\sigma_i\otimes\sigma_j) J_\mc{N}]$.  Then
\begin{equation}
\cv(\mc{N})=1+\sigma_{\max}(A)
\end{equation}
where $\sigma_{\max}(A)$ is the largest singular value of $A$.
\end{theorem}

\begin{remark}
For a unital channel $\mc{N}$, the Bloch vector $\mbf{c}$ is zero.  Since we are able to obtain the largest value of $\gamma_k\mbf{b}_k^TA\mbf{a}_k$ for each value $k$, it follows that
\begin{equation}
\label{Eq:unital-GM-tight}
\cv(\mc{N})=2\Lambda^2(J_{\mc{N}});
\end{equation}
i.e. the upper bound in Eq. \eqref{Eq:cv-GM} is tight.
\end{remark}

\medskip

\noindent \textit{Example: Pauli Channels.}  As a nice example of Theorem \ref{Thm:qubit-cv}, consider the family of Pauli channels, which consists of any qubit channel having the form
\begin{equation}
\mc{N}(\rho)=p_0 \rho+ p_1X\rho X+ p_2 Z\rho Z +p_3 Y\rho Y,
\end{equation}
where $\{X,Y,Z\}$ are the standard Pauli matrices and $\sum_{i=0}^3p_i=1$.  We can write the Choi matrix as
\begin{equation}
J_{\mc{N}}=2\sum_{i=0}^3p_i\op{\Phi^+_i}{\Phi^+_i},
\end{equation}
where $\ket{\Phi^+_i}$ denotes the four Bell states.  It is easy to see that the correlation matrix of $J_{\mc{N}}$ is diagonal with entries $\{p_0+p_1-p_2-p_3, -p_0+p_1-p_2+p_3, p_0-p_1-p_2+p_3\}$.  Therefore, by Theorem \ref{Thm:qubit-cv} we can conclude that 
\begin{align}
\label{Eq:cv-Pauli}
\cv(\mc{N})=2(p^\downarrow_{3}+p^\downarrow_{2}),
\end{align}
where $p^\downarrow_{3}$ and $p^\downarrow_{2}$ are the two largest probabilities of Pauli gates.

Notice that $\cv(\mc{N})$ will equal its largest value of two if and only if there are no more than two Pauli gates applied with nonzero probability in $\mN$.  In particular, when $p^\downarrow_{3}=p^\downarrow_{2}=\frac{1}{2}$ the channel is entanglement-breaking; in fact it is a classical channel.  Hence, this example shows that a channel's communication value captures a property distinct from its ability to transmit entanglement.

\subsection{Werner-Holevo Channels}

\label{Sect:Werner-Holevo}

The Werner-Holevo family of channels \cite{Werner-2002a,Leung-2015a} is defined by
\begin{align}\label{eqn:WH-defn}
\mathcal{W}_{d,\lambda} := \lambda \Phi_{0}(X) + (1-\lambda) \Phi_{1}(X) \ ,
\end{align}
where 
\begin{align*}
\Phi_{0}(X) &= \frac{1}{n+1} \left((\tr(X))\mbb{I} + X^{\T} \right) \\
\Phi_{1}(X) &=\frac{1}{n-1}\left((\tr(X))\mbb{I} - X^{\T} \right) \ .
\end{align*}
This implies the Choi matrix is given by 
\begin{align}
    J_{\mc{W}_{d,\lambda}} =\lambda \frac{2}{d+1}\Pi_+ + (1-\lambda)\frac{2}{d-1}\Pi_-,
\end{align}
where $\Pi_+ = \frac{1}{2}\left(\mbb{I} + \mbb{F}\right)$ and $\Pi_- = \mbb{I} - \Pi_+$.



\begin{proposition}\label{prop:cv-WH-1-copy}
    The communication value of the Werner-Holevo channel is given by 
    \begin{align*}
       \cv(\mc{W}_{d,\lambda})
        =
        \begin{cases}
            \frac{d(d+1-2\lambda)}{d^{2}-1}  & \lambda \leq \frac{1+d}{2d} \\
            \frac{2d\lambda}{1+d} & \lambda > \frac{1+d}{2d}
        \end{cases}
    \end{align*}
\end{proposition}
\begin{proof}
Let $\mc{U}(\mbb{C}^d)$ denote the group of $d\times d$ unitary operators.  Since $J_{\mc{W}_{d,\lambda}}$ enjoys $U\otimes U$ invariance under conjugation for every $U\in\mc{U}(\mbb{C}^d)$ \cite{Werner-1989a}, we can apply the ``twirling map''
\begin{equation}
\label{Eq:UU-twirling}
    \mc{T}_{UU}(X)=\int_{\mc{U}} (U\otimes U) X(U\otimes U)^\dagger d U
\end{equation}
to the $A'B'$ systems of $\sigma$ while leaving the cv invariant:
\[\tr[ \mc{T}_{UU}(J_{\mc{W}}) \Omega^{AB} ] = \tr[J_{\mc{W}} \mc{T}_{UU}(\Omega^{AB})].\]   Furthermore, since $\mc{T}_{UU}$ preserves the constraints on $\Omega^{AB}$, we can conclude without loss of generality the optimizer is given by $X := \mc{T}_{UU}(\Omega^{AB}) = x\mbb{I}^{AB} + y \mathbb{F}$ for some choice of $x,y$, i.e. it is an element of the $U\otimes U$-invariant space of operators. 

As the space of PPT and SEP $UU$-invariant operators are the same \cite{Vollbrecht-2001a}, we can relax the optimization program to only requires $X \in \ppt$. As is shown in \eqref{eqn:PPTPrimal}, this means we require $X$ satisfies $X \geq 0, \Gamma^{B}(X) \geq 0, \tr_{A}[\Omega] = \mbb{I}^{B}$. We will therefore convert these linear constraints into linear constraints on $x,y$. 

Note that $\{\Pi_{+},\Pi_{-}\}$ define an orthogonal basis for the space spanned by $\{\mbb{I},\mbb{F}\}$. Therefore we can write $X = (x+y)\Pi_{+} + (x-y)\Pi_{-}$, and the positivity constraints on $X$ are given by
\begin{align}
    x \pm y \geq 0 \ .
\end{align}
Similarly, $\Gamma^{B}(X) = x\mbb{I}^{AB} + y\Phi^{+}$, where $\Phi^{+}$ is the unnormalized maximally entangled state. An orthogonal basis for the space spanned by $\{\mbb{I}^{AB},\Phi^{+}\}$ is given by $\{\Phi^{\perp} := d\mbb{I}-\Phi^{+},\Phi^{+}\}$. Therefore, $X^{\Gamma^{B}} = \frac{x}{d}(\Phi^{\perp}+\Phi^{+}) + y\Phi^{+}$. It follows the partial transpose positivity constraints simplify to
\begin{align}
    x \geq 0 \hspace{1cm} x+yd \geq 0
\end{align}

The objective function is given by
\begin{align}
    & \tr[J_{\mN}X] \notag \\
    = & \tr\left[\left(\frac{2\lambda}{d+1}\Pi_+ + \frac{2(1-\lambda)}{d-1}\Pi_{-}\right)X\right] \notag \\
    = & \lambda d (x+y) + (1-\lambda)d(x-y) \notag \\
    = & d(x+(2\lambda - 1)y) \ .
\end{align}
Lastly, the trace condition is given by
\begin{align}
     x\tr_{A}[\mbb{I}^{AB}] + y\tr_{A}[\mbb{F}] = \mbb{I}^{B}
     \Rightarrow xd + y = 1 \ .
\end{align}
Combining these, we have the linear program
\begin{align}
    \text{maximize} \quad & d(x+(2\lambda - 1)y) \notag \\
    \text{subject to} \quad & xd + y = 1 \label{eqn:WH-LP-trace} \\
    & x+yd \geq 0 \label{eqn:WH-LP-3} \\
    & x+y \geq 0 \label{eqn:WH-LP-2} \\
    & x-y \geq 0 \label{eqn:WH-LP-1} \\
    & x \geq 0 \label{eqn:WH-LP-4} \ .
\end{align}
\eqref{eqn:WH-LP-trace} implies $y = 1 -xd$. \eqref{eqn:WH-LP-3},\eqref{eqn:WH-LP-1} imply $\frac{1}{d+1} \leq x \leq \frac{d}{d^{2}-1}$. $\eqref{eqn:WH-LP-2}$ implies $x \leq \frac{1}{d-1}$, which is always satisfied if $x \leq \frac{d}{d^{2}-1}$. Finally, if $x$ satisfies these constraints,
$$ x+yd = x + (1-xd)d \geq \frac{d}{d+1} + d - \frac{d^{2}}{d+1} \geq 0 \ , $$
so \eqref{eqn:WH-LP-4} is also satisfied. Thus we have reduced the LP to
\begin{align}
    \text{maximize} \quad & d(x+(2\lambda - 1)(1-xd)) \label{eqn:WH-LP-simplified-obj} \\
    \text{subject to} \quad & \frac{1}{d+1} \leq x \leq \frac{1}{d-1} \notag \ .
\end{align}
Taking the derivative of the objective function one finds that for $\lambda \leq \frac{1+d}{2d}$, the derivative is positive. Therefore, 
\begin{align*}
x^{*} = 
    \begin{cases} 
        \frac{d}{d^{2}-1} & \lambda \leq \frac{1+d}{2d} \\
        \frac{1}{d+1} & \text{ otherwise} \ .
    \end{cases}
\end{align*}
Plugging $x^{*}$ into \eqref{eqn:WH-LP-simplified-obj} completes the proof.
\end{proof}
\noindent In Section \ref{sec:relaxations_of_cv}, we generalize this derivation to determine the PPT relaxation of cv for the $n$-fold Werner-Holevo channels.

\section{Multiplicativity of cv }

\label{Sect:Multiplicativity}

We next consider how the communication value behaves when we combine two or more channels.  The cv is multiplicative for two channels $\mc{N}$ and $\mc{M}$ if
\begin{equation}
\cv(\mc{N}\otimes\mc{M})=\cv(\mc{N})\cv(\mc{M}).
\end{equation}
When multiplicativity holds, it means that an optimal strategy for guessing channel inputs involves using uncorrelated inputs and measurements across the two channels.  A concrete example of non-multiplicativity is given by the Holevo-Werner family of channels, as proven in Section \ref{Sect:non-multiplicativity}.  In general, it is a hard problem to decide whether two channels have a multiplicative communication value.  More progress can be made when relaxing this problem to the PPT cone, and we conduct such an analysis in Section \ref{sec:relaxations_of_cv}.  Here we resolve on the question of multiplicativity for a few special cases.

\subsection{Entanglement-Breaking Channels}

Our first result shows that non-multiplicativity arises only if the channel is capable of transmitting entanglement.

\begin{theorem}
\label{Thm:EBC-mult}
If $\mc{N}$ is an entanglement-breaking channel, then $\cv(\mc{N}\otimes\mc{M})=\cv(\mc{N})\cv(\mc{M})$ for an arbitrary channel $\mc{M}$.
\end{theorem}
\begin{proof}
Since $\mN$ is EB, its Choi matrix has the form $J_{\mN}=\sum_x\Pi_x\otimes\rho_x$ for some POVM $\{\Pi_x\}_x$.  The dual optimization of cv (i.e. Eq. \eqref{Eq:cv-dual}) can then be expressed as
\begin{align}
\label{Eq:EB-multiplicative} 
&\cv(\mc{N}\otimes\mc{M})=\min \;\;\tr[Z^{BB'}]\notag\\
&\text{subject to}\;\; Z^{BB'}\geq \sum_x\tr[\Pi_x\rho^A]\op{x}{x}\otimes\mc{M}(\sigma_x) \notag\\
& \qquad\sigma_x:=\tr_A[(\Pi^A_x\otimes\mbb{I}^{A'})\rho^{AA'}]/\tr[\Pi_x\rho^A],\;\; \forall \rho^{AA'}.
\end{align}
Suppose that $Z^B\geq \mc{N}(\rho)$ for all $\rho$ and $Z^{B'}\geq\mc{M}(\sigma)$.  Then
\begin{align}
    \sum_x\tr[\Pi_x\rho^A]\op{x}{x}\otimes\mc{M}(\sigma_x)&\leq  \sum_x\tr[\Pi_x\rho^A]\op{x}{x}\otimes Z^{B'}\notag\\
    &\leq Z^B\otimes Z^{B'}.\notag
\end{align}
Thus $Z^B\otimes Z^{B'}$ is feasible in Eq. \eqref{Eq:EB-multiplicative}.  By choosing $Z^B$ and $Z^{B'}$ to be the dual optimizers for $\mc{N}$ and $\mc{M}$, respectively, we have $\cv(\mc{N}\otimes\mc{M})\leq \cv(\mc{N})\cv(\mc{M})$.  Since the opposite inequality trivially holds, we have the equality $\cv(\mc{N}\otimes\mc{M})= \cv(\mc{N})\cv(\mc{M})$.
\end{proof}
\noindent By applying Theorem \ref{Thm:EBC-mult} iteratively across $n$ copies of an entanglement-breaking channel, we obtain a single-letter formulation of the cv capacity.
\begin{corollary}
If $\mN$ is an entanglement-breaking channel, then
\begin{equation}
    \mathcal{CV}(\mN)=\cv(\mN).
\end{equation}
\end{corollary}

\subsection{Covariant Channels}

We next turn to channels that have a high degree of symmetry.  To study the question of multiplicativity, it will be helpful to use the relationship between cv and GME.  The following is a powerful result proven in Ref. \cite{Zhu-2011a} regarding multiplicativity of the GME.  We say an operator $\rho^{AB}$ is component-wise non-negative if there exists an orthonormal product basis $\{\ket{i,j}\}_{i,j}$ such that $\bra{i,j}\rho\ket{i',j'}\geq 0$ for all $i,j,i',j'$.
\begin{lemma}[\cite{Zhu-2011a}]
\label{Lemma:non-negative-multiplicative}
If $\rho^{AB}$ is component-wise non-negative and $\sigma^{A'B'}$ is any other density operator, then $\Lambda^2(\rho\otimes\sigma)=\Lambda^2(\rho)\Lambda^2(\sigma)$.
\end{lemma}

\noindent 
For example, the Choi matrix of the identity channel, $\phi^+_{d'}=J_{\id_{d'}}$, is component-wise non-negative in the computational basis.  Therefore by the previous lemma we have \[\Lambda(J_{\mN}\otimes\id_{d'})=\Lambda(J_{\mN})\Lambda(\id_{d'})=\Lambda(J_{\mN})\]
for any channel $\mN$.  As we will now show, this sort of multipicativity can be readily extended to the communication value for channels with symmetry.  

Let $\mc{G}$ be any group with an irreducible unitary representation on $\mbb{C}^d$.  Then, as we did in Eq. \eqref{Eq:UU-twirling}, let $\mc{T}_{UU}$ denote the bipartite group twirling map with respect to $\mc{G}$,
\begin{align}
    \mc{T}_{UU}(\rho^{AB})=\int_{\mc{G}}dU (U\otimes U)\rho(U\otimes U)^\dagger.
\end{align}
A channel $\mc{N}$ is called $\mc{G}$-covariant if $\mc{N}(U_g\rho U_g^\dagger)=U_g\mc{N}(\rho)U_g^\dagger$ for all $g\in\mc{G}$ and all $\rho$.  On the level of Choi matrices, this is equivalent to $\mc{T}_{\ol{U}U}(J_\mc{N})=J_{\mc{N}}$, where $\ol{U}$ denote complex conjugation.  Note that $\Gamma_A\circ\mc{T}\circ\Gamma_A$ is the CPTP twirling map $\mc{T}_{\ol{U}U}$, where $\Gamma_A$ is the partial transpose on system $A$.    Likewise, the map $\Gamma\circ\mc{N}\circ\Gamma$ is $\mc{G}$-covariant if $\mc{T}_{UU}(J_\mc{N})=J_{\mc{N}}$, where $\Gamma$ denotes the transpose map.

\begin{theorem}\label{thm:covariance-multiplicativity}
Let $\mc{G}$ and $\mc{G}'$ have irreducible unitary representations on $\mbb{C}^d$ and $\mbb{C}^{d'}$ respectively.  Suppose either $\mN$ or $\Gamma\circ\mN\circ\Gamma$ is $\mc{G}$-covariant and likewise either $\mN'$ or $\Gamma\circ\mN'\circ\Gamma$ is $\mc{G}$'-covariant.  Further suppose that $J_{\mN}$ is component-wise non-negative.  Then
\begin{align}
    \cv(\mN\otimes\mN')=\cv(\mN)\cv(\mN').
\end{align}
\end{theorem}
\begin{proof}
Let $\op{\alpha}{\alpha}^A\otimes\op{\beta}{\beta}^B$ be a product operator with trace equaling $d$ and satisfying 
$d\Lambda^2(J_{\mc{N}})=\bra{\alpha,\beta}J_{\mc{N}}\ket{\alpha,\beta}$.  Suppose now that either $\mc{N}$ or $\Gamma\circ\mc{N}\circ\Gamma$ is $\mc{G}$-covariant.  In either case we have
\begin{align}
    \bra{\alpha,\beta}J_{\mc{N}}\ket{\alpha,\beta}&=\bra{\alpha,\beta}\mc{T}(J_{\mc{N}})\ket{\alpha,\beta}\notag\\
    &=\tr\left[J_{\mc{N}}\mc{T}^\dagger\left(\op{\alpha,\beta}{\alpha,\beta}\right)\right]\notag\\
    &=\tr[J_\mc{N}\Omega^{AB}],
\end{align}
where $\Omega^{AB}=\mc{T}^\dagger(\op{\alpha,\beta}{\alpha,\beta})$.  Note that $\Omega^{AB}$ has trace equaling $d$, and since $U_g$ is an irrep, we have $\tr_{A}\Omega^{AB}=\mbb{I}$.  Hence $\cv(\mc{N})\geq d\Lambda^2(J_{\mc{N}})$, and an analogous argument for $\mN'$ establishes that $\cv(\mN')\geq d'\Lambda^2(J_{\mN'})$.  Therefore,
\begin{align}
    \cv(\mN\otimes\mN')&\geq \cv(\mN)\cv(\mN')\notag\\
    &\geq dd'\Lambda^2(J_{\mc{N}})\Lambda^2(J_{\mN'})\notag\\
    &= dd'\Lambda^2(J_{\mN}\otimes J_{\mN'}),
\end{align}
where the last equality follows from Lemma \ref{Lemma:non-negative-multiplicative}.  However, by the upper bound in Eq. \eqref{Eq:cv-GM} this inequality must be tight, which implies the desired multiplicativity.

\end{proof}

Using this theorem, we can compute the cv capacity for certain channels.
\begin{corollary}
\label{Cor:CV-capacity-symmetry}
Let $\mc{G}$ have irreducible unitary representation on $\mbb{C}^d$.  Suppose that either $\mN$ or $\Gamma\circ\mN\circ\Gamma$ is $\mc{G}$-covariant and $J_{\mN}$ is component-wise non-negative.  Then
\begin{align}
    \mc{CV}(\mN)=\cv(\mN).
\end{align}
\end{corollary}
\begin{proof}
It suffices to prove that $\cv(\mN^{\otimes n})=n\cv(\mN)$.  This follows directly from Theorem \ref{thm:covariance-multiplicativity} by letting $\mN'=\mN^{\otimes n-1}$ and $\mc{G}'=\mc{G}^{\otimes n}$.
\end{proof}

For example, in a qubit system all Pauli channels satisfy the conditions of Corollary \ref{Cor:CV-capacity-symmetry}.  These are channels of the form
\begin{align}
    \mN_{\text{Pauli}}(\rho)=p_0\rho+p_1\sigma_x\rho\sigma_x+p_2\sigma_z\rho\sigma_z+p_3\sigma_y\rho\sigma_y,
\end{align}
and they are covariant with respect to the Pauli group.  Moreover, $J_{\mN_{\text{Pauli}}}$ can always be converted into a matrix with non-negative entries by local unitaries.  Hence using Eq. \eqref{Eq:cv-Pauli} we have
\begin{equation}
    \mc{CV}(\mN_{\text{Pauli}})=2(p_3^{\downarrow}+p_2^{\downarrow}).
\end{equation}

As another example, consider the $d$-dimensional partially depolarizing channel $\mc{D}_{d,\lambda}$ given by
\begin{equation}
    \mc{D}_{d,\lambda}(\rho):=\lambda\rho+(1-\lambda)\frac{\mbb{I}}{d}, \qquad 0\leq\lambda\leq 1.
\end{equation}
The channel $\Gamma\circ\mc{D}_{d,\lambda}\circ\Gamma$ is $\mc{G}$-covariant with respect to the full unitary group on $\mbb{C}^d$ \cite{Horodecki-1999a}.  The Choi matrix is given by $J_{\mc{D}_{d,\lambda}}=\lambda\phi^+_d+(1-\lambda)\mbb{I}\otimes\mbb{I}/d$, which is clearly component-wise non-negative.  Thus by Corollary \ref{Cor:CV-capacity-symmetry} we have
\begin{align}
    \mc{CV}(\mc{D}_{d,\lambda})=\cv(\mc{D}_{d,\lambda})=\lambda d+(1-\lambda).
\end{align}

We remark that the Werner-Holevo family of channels introduced in Section \ref{Sect:Werner-Holevo} fail to satisfy Corollary \ref{Cor:CV-capacity-symmetry} since they are not component-wise non-negative.  In fact, their cv is non-multiplicative, as we will see below.  Nevertheless, Theorem \ref{thm:covariance-multiplicativity} can be applied to a Werner-Holevo channel by using in parallel with another channel that is component-wise non-negative.  For example, when trivially embedding $\mc{W}_{d,\lambda}$ into a larger system we have multiplicativity:
\begin{equation}
    \cv(\mc{W}_{d,\lambda}\otimes\id_{d'})=d'\cv(\mc{W}_{d,\lambda}).
\end{equation}
This result is perhaps surprising since mutliplicativity appears to not hold when we relax the cv optimization to the cone of PPT operators, as is shown in Section \ref{subsec:WHID} (See Fig. \ref{fig:cv-ppt-versus-cv-wh-with-id}).
\medskip

\subsection{Qubit Channels}

\label{Sect:qubits-multiplicative}

In Section \ref{Sect:qubits} we derived an explicit formula for the communication value of qubit channels.  Namely, $\cv(\mN)=1+\sigma_{\max}(N)$, where $N$ is the correlation matrix of $J_{\mN}$.  Here we show that the cv is multiplicative when using two qubit channels in parallel.
\begin{theorem}
\label{Thm:qubit-multiplicative}
Suppose $\mc{M},\mc{N}\in\cptp(A\to B)$ with $d_A=d_B=2$.  Then
\begin{equation}
\cv(\mc{M}\otimes\mc{N})=\cv(\mc{M})\cv(\mc{N}).
\end{equation}
\end{theorem}

\begin{proof}
For $\mc{M}$ and $\mc{N}$ consider their Choi matrices
\begin{align}
J_{\mc{M}}&=\frac{1}{2}\bigg(\mbb{I}\otimes(\mbb{I}+\mbf{c}\cdot\vec{\sigma})+\sum_{i}m_i\sigma_i\otimes\sigma_i\bigg),\notag\\
J_{\mc{N}}&=\frac{1}{2}\bigg(\mbb{I}\otimes(\mbb{I}+\mbf{d}\cdot\vec{\sigma})+\sum_in_i\sigma_i\otimes\sigma_i\bigg).
\end{align}
Notice that their correlation matrices $M=\text{diag}[m_1,m_2,m_3]$ and  $N=\text{diag}[n_1,n_2,n_3]$ are diagonal.  An arbitrary channel can always be converted into this form by performing appropriate pre- and post $SU(2)$ rotations on the channel, which do not change the communication value.  Define the operator 
\begin{align}
\label{Eq:qubit-multiplicative-dual-1}
Z^{BB'}=\frac{1}{4}\bigg(&(\mbb{I}+\mbf{c}\cdot\vec{\sigma})\otimes(\mbb{I}+\mbf{d}\cdot\vec{\sigma})+(\sigma_{\max}(M)\notag\\
&+\sigma_{\max}(N)+\sigma_{\max}(M)\sigma_{\max}(N))\mbb{I}\otimes\mbb{I}\bigg).
\end{align}
Since $\tr[Z^{BB'}]=\cv(\mc{M})\cv(\mc{N})$, by the dual characterization of cv given in Eq. \eqref{Eq:cv-dual}, we will prove that $\cv(\mc{M}\otimes\mc{N})\leq\cv(\mc{M})\cv(\mc{N})$ if we can show that
\begin{equation}
\label{Eq:qubit-multiplicative-dual-2}
Z^{BB'}\geq \mc{M}\otimes\mc{N}(\rho^T).
\end{equation}
for an arbitrary two-qubit state
\begin{align}
\rho^{AA'}=\frac{1}{4}\bigg(\mbb{I}\otimes\mbb{I}+\mbf{r}\cdot\vec{\sigma}\otimes\mbb{I} +\mbb{I}\otimes\mbf{s}\cdot\vec{\sigma}+\sum_{i,j}c_{ij}\sigma_i\otimes\sigma_j\bigg).\notag
\end{align}
Note that we have the action $\mc{M}(\mbb{I})=\mbb{I}+\mbf{c}\cdot\vec{\sigma}$, $\mc{M}(\sigma_x)=m_x\sigma_x$, $\mc{M}(\sigma_y)=-m_y\sigma_y$, $\mc{M}(\sigma_z)=m_z\sigma_z$, and likewise for the action of $\mc{N}$.  Hence
\begin{align}
\mc{M}\otimes\mc{N}(\rho^T)=\frac{1}{4}&\bigg((\mbb{I}+\mbf{c}\cdot\vec{\sigma})\otimes(\mbb{I}+\mbf{d}\cdot\vec{\sigma})+M\mbf{r}\cdot\vec{\sigma}\otimes\mbb{I}\notag\\
&+\mbb{I}\otimes N\mbf{s}\cdot\vec{\sigma}+\sum_{i,j}m_in_jc_{ij}\sigma_i\otimes\sigma_j\bigg).
\end{align}
When comparing with Eq. \eqref{Eq:qubit-multiplicative-dual-1}, we see that Eq. \eqref{Eq:qubit-multiplicative-dual-2} reduces to
\begin{align}
\label{Eq:qubit-multiplicative-dual-3}
\Vert &M\mbf{r}\cdot\vec{\sigma}\otimes\mbb{I}+\mbb{I}\otimes N\mbf{s}\cdot\vec{\sigma}+\sum_{i,j}m_in_jc_{ij}\sigma_i\otimes\sigma_j\Vert_\infty\notag\\
&\leq  \sigma_{\max}(M)+\sigma_{\max}(N)+\sigma_{\max}(M)\sigma_{\max}(N).
\end{align}
To prove this inequality, we note that effectively the non-unital components of $\mc{M}$ and $\mc{N}$ do not appear here.  That is, let $\wt{\mc{M}}$ and $\wt{\mc{N}}$ be the unital CPTP maps defined by the Choi matrices
\begin{align}
J_{\wt{\mc{M}}}&=\frac{1}{2}\bigg(\mbb{I}\otimes\mbb{I}+\sum_{i}m_i\sigma_i\otimes\sigma_i\bigg),\notag\\
J_{\wt{\mc{N}}}&=\frac{1}{2}\bigg(\mbb{I}\otimes\mbb{I}+\sum_in_i\sigma_i\otimes\sigma_i\bigg).
\end{align}
Letting $\ket{\psi}$ denote an eigenvector of largest eigenvalue for the operator on the LHS of Eq. \eqref{Eq:qubit-multiplicative-dual-3}, we have
\begin{align}
\Vert \mbb{I}\otimes\mbb{I}+M\mbf{r}\cdot\vec{\sigma}\otimes&\mbb{I}+\mbb{I}\otimes N\mbf{s}\cdot\vec{\sigma}+\sum_{i,j}m_in_jc_{ij}\sigma_i\otimes\sigma_j\Vert_\infty\notag\\
&=4\bra{\varphi}\wt{\mc{M}}\otimes\wt{\mc{N}}(\rho^T)\ket{\varphi}\notag\\
&=4\tr\left[\rho^{AA'}\otimes\op{\varphi}{\varphi}^{BB'}J_{\wt{\mc{M}}}\otimes J_{\wt{\mc{N}}}\right]\notag\\
&\leq 4\Lambda^2(J_{\wt{\mc{M}}}\otimes J_{\wt{\mc{N}}})\notag\\
&=2\Lambda^2(J_{\wt{\mc{M}}})\cdot 2\Lambda^2(J_{\wt{\mc{N}}})\notag\\
&=(1+\sigma_{\max}(M))(1+\sigma_{\max}(N)),
\end{align}
where we have used the fact that the GME is multiplicative for unital qubit channels (Lemma \ref{Lemma:non-negative-multiplicative}), along with Eq. \eqref{Eq:unital-GM-tight}.  This proves Eq. \eqref{Eq:qubit-multiplicative-dual-3}.

\end{proof}

A natural question is whether Theorem \ref{Thm:qubit-multiplicative} can generalized to the case in which only one of the channels is a qubit channel.  Unfortunately the proof of Theorem \ref{Thm:qubit-multiplicative} relies heavily on the Pauli representation of qubit channels, and we therefore only conjecture that qubit channels possess an even stronger form of multiplicativity.

\medskip
\noindent\textit{Conjecture.}  If $\mM\in\cptp(A\to B)$ with $d_A=d_B=2$, then
\[\cv(\mM\otimes\N)=\cv(\mM)\cv(\mN)\]
for any other channel $\mN$.

\subsection{Non-Multiplicativity in Qutrits}

\label{Sect:non-multiplicativity}

In the previous sections we identified examples of channels for which the communication value is multiplicative.  We now provide an example of channels that demonstrate non-multiplicativity.  Our construction is the Werner-Holevo channels, which were previously known to exemplify non-additivity of a channel's minimum output purity \cite{Werner-2002a}, \cite{Zhu-2011a}.  Specifically, the channel $\mc{W}_{d,0}$ has a Choi matrix proportional to the anti-symmetric subspace projector,
\begin{equation}\label{eqn:anti-sym-choi}
J_{\mc{W}_{d,0}}=\frac{1}{d-1}(\mbb{I}\otimes\mbb{I}-\mbb{F}).
\end{equation}
The entanglement properties of this operator have been well-studied \cite{Christandl-2012a, Hubener-2009a, Zhu-2011a}.  In particular, Zhu \textit{et al.} have computed its one and two-copy geometric measures of entanglement to be
\begin{align}
\max_{\ket{\alpha}^{A}\ket{\beta}^B}\bra{\alpha,\beta}J_{\mc{W}_{d,0}}\ket{\alpha,\beta}&=\frac{1}{d-1} \label{Eq:asym-projGM}\\
\max_{\ket{\alpha}^{AA'}\ket{\beta}^{BB'}}\bra{\alpha,\beta}J_{\mc{W}_{d,0}}\otimes J_{\mc{W}_{d,0}}\ket{\alpha,\beta}&=\frac{2}{d(d-1)}.\label{Eq:asym-projGM2}
\end{align}
Equation \eqref{Eq:asym-projGM2} is strictly larger than the square of  Eq. \eqref{Eq:asym-projGM} whenever $d\geq 3$.  
Furthermore, the maximization in Eq. \eqref{Eq:asym-projGM2} is attained whenever $\ket{\alpha}^{AA'}$ and $\ket{\beta}^{BB'}$ are maximally entangled states.  Thus, we consider the separable operator
\begin{equation}
\sigma^{AA':BB'}=\sum_{k=1}^{d^2}\op{\varphi_k^+}{\varphi_k^+}^{AA'}\otimes\op{\varphi_k^+}{\varphi_k^+}^{BB'},
\end{equation}
where $\{\ket{\varphi_k^+}\}_{k=1}^{d_2}$ is an orthonormal basis consisting of maximally entangled states for $\mbb{C}^d\otimes\mbb{C}^d$.  This satisfies the conditions of Proposition \ref{Prop:Proposition-cv}.  Hence, we conclude
\begin{align}
\cv(\mc{W}_{d,0})&=\frac{d}{d-1} \label{eqn:anti-sym-cv} \\
\cv(\mc{W}_{d,0}\otimes\mc{W}_{d,0})&=\frac{2d}{d-1},
\end{align}
which yields $\cv(\mc{W}_{d,0})^2<\cv(\mc{W}_{d,0}\otimes\mc{W}_{d,0})$ when $d\geq 3$.  Most notably, for $d=3$ we have $\cv(\mc{W}_{d,0}\otimes\mc{W}_{d,0})=3$ while $\cv(\mc{W}_{d,0})^2=2.25$.

\begin{remark}
In Section \ref{sec:relaxations_of_cv}, we use that the spaces of PPT and SEP $UUVV$-invariant operators are equivalent \cite{Vollbrecht-2001a}, to determine the range of $\lambda$ that satisfy non-multiplicativity for $\cv(\mc{W}_{d,\lambda} \otimes \mc{W}_{d,\lambda})$ numerically.
\end{remark}

\section{Entanglement-Assisted CV}
\label{sec:entanglement_assisted_cv}

\begin{figure}[b]
  \includegraphics[width=\linewidth]{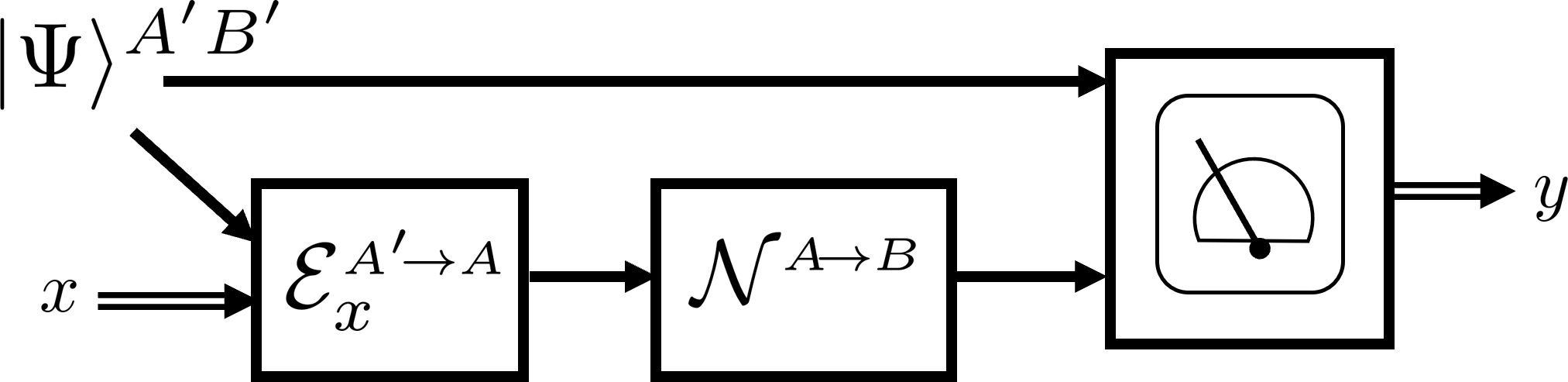}
  \caption{Entanglement-assisted communication value scenario.}
  \label{fig:ea_cv}
\end{figure}

We next generalize the communication scenario and allow the sender and receiver to share entanglement.  Remarkably, this added resource simplifies the problem immensely.  In what follows, we will allow Alice and Bob to share an entangled state $\varphi^{A'B'}$ that can be used to increase the channel cv.  The most general entanglement-assisted protocol is as follows (see Fig. \ref{fig:ea_cv}).  For input $x\in[n]$, Alice performs a CPTP map $\mc{E}_x\in\cptp(A'\to A)$ on her half of the entangled state $\varphi^{A'B'}$.  System $A$ is then fed into the channel, and Bob finally performs a POVM $\{\Pi_y^{BB'}\}_{y\in[n']}$ on systems $BB'$.  The induced channel has transition probabilities given by
\begin{equation}
\label{Eq:Ea-cv}
P(y|x)=\tr\left[\Pi_y^{BB'}\left(\mc{N}^{A\to B}\circ\mc{E}_x^{A'\to A}\otimes\id^{B'}[\varphi^{A'B'}]\right)\right].
\end{equation}
Note that this scenario corresponds to the one used in superdense coding \cite{Bennett-1992a}.  The entanglement-assisted channel cv can now be defined.
\begin{definition}
The entanglement-assisted communication value (ea cv) of a quantum channel $\mc{N}\in\cptp(A\to B)$, denoted $\cv^{*}(\mN)$, is 
\begin{align}
    \sup_{\varphi^{A'B'}}\max\{\cv(\mbf{P})\;|\;\text{$P(y|x)$ given by Eq. \eqref{Eq:Ea-cv}}\} \, ,
\end{align}
where the supremum is taken over all entangled states $\varphi^{A'B'}$ (and all dimensions $A'B'$), while the maximization considers all $n,n'\in\mbb{N}$ along with arbitrary states $\{\rho_x\}_{x\in[n]}$ and POVMs $\{\Pi_y\}_{y\in[n']}$.
\end{definition}
\begin{theorem}
\label{Thm:ea-cv}
For an arbitrary channel $\mc{N}\in\cptp(A\to B)$,
\begin{align}
\cv^*(\mc{N})=\max &\;\;\tr[\sigma^{AB}J_\mc{N}]\notag\\
&\;\;\tr_A[\sigma^{AB}]=\mbb{I}^B\notag\\
&\;\;\sigma^{AB}\in\text{Pos}(A:B),
\label{Eq:theorem-ea-cv}
\end{align}
where $\text{Pos}(A:B)$ denotes the positive cone on $AB$.  Moreover, $\cv^*(\mc{N})$ is attained using a $(d_A)$-dimensional maximally entangled state.
\end{theorem}
\noindent In other words, the restriction of $\sigma^{AB}$ to the separable cone (\textit{cf} Eq. \eqref{Eq:Proposition-cv}) is removed when considering the entanglement-assisted problem.
\begin{proof}
It is clear that we need to only consider pure states in the supremum since $\cv(\mbf{P})$ is convex-linear w.r.t. $\varphi^{A'B'}$.  Let $\ket{\varphi}^{A'B'}$ be arbitrary.  We first show that without loss of generality we can take $\ket{\varphi}$ to be maximally entangled.  Recall Nielsen's Theorem \cite{Nielsen-1999a} (see also \cite{Lo-2001a}), which ensures the existence of an LOCC transformation $\ket{\Phi^+_{d_{A'}}}^{A'\tilde{A'}}\to\ket{\varphi}^{A'B'}$, where $\ket{\Phi^+_{d_{A'}}}^{A'\tilde{A'}}$ is a maximally entangled state on $A'\tilde{A'}$, with $\tilde{A'}\cong A'$.  Explicitly, there exists a measurement on Bob's side with Kraus operators $\{M_k\}_k$ and correcting unitaries $\{U_k\}_k$ on Alice's side such that $U^{A'}_k\otimes M^{\tilde{A'}\to B'}_k\ket{\Phi^+_{d_{A'}}}=\sqrt{p(k)}\ket{\varphi_{k}}$. Using that cv is achieved by minimal error discrimination, i.e. $\cv^*(\mathbb{P}) = \sum_{x} P(x|x)$,
\begin{align}
\label{Eq:Ea-cv-1}
& \cv^*(\mbf{P})\notag\\
=&\sum_{x,k}\tr\left[\Pi_x^{BB'}\left(\mc{N}^{A\to B}\circ\mc{E}_x^{A'\to A}\otimes\id^{B'}\left[p(k)\op{\varphi_{k}}{\varphi_{k}}\right]\right)\right].
\end{align}
where by construction
$$ p(k)\op{\varphi_{k}}{\varphi_{k}} = [(U_k\otimes M_k)\Phi^{+A'\tilde{A'}}(U_k\otimes M_k)^\dagger] \ . $$
Notice that $\{(\mbb{I}\otimes M_k)^\dagger\Pi_x(\mbb{I}\otimes M_k)\}_{k,x}$ constitutes a set of POVM elements on $B\tilde{A'}$.  This follows from the fact that $\{M_k\}_k$ are Kraus operators for a CPTP map, and so the dual of this map, $X\to \sum_kM_k^\dagger(X)M_k$ , is unital.  Likewise, letting $\mc{U}_k(\cdot):=U_k(\cdot) U_k^\dagger$ denote a unitary channel, the collection $\{\mc{E}_x\circ \mc{U}_k\}_{x,k}$ forms a family of encoding maps.  Therefore, we can express Eq. \eqref{Eq:Ea-cv-1} as 
\begin{align}
\label{Eq:Ea-cv-2}
\cv^*(\mbf{P})&=\sum_z\tr\!\left[\hat{\Pi}_z^{B\tilde{A'}}\!\left(\mc{N}^{A\to B}\circ\hat{\mc{E}}_z^{A'\to A}\otimes\id^{\tilde{A'}}[\Phi^{+A'\tilde{A'}}]\right)\right],
\end{align}
where the $\hat{\mc{E}}_z$ and $\hat{\Pi}_z$ are the concatenated encoders and decoder.  This shows that we can restrict attention just to shared maximally entangled states.  Furthermore, without loss of generality, we can assume that $d_{A'}\geq d_{A}$.  The reason is that the transformation $\ket{\Phi^+_{d_{A''}}}^{A''\tilde{A''}}\to\ket{\varphi}^{A'B'}$ is always possible for any $d_{A''}\geq d_{A'}$; so we could have just as well used the same argument with system $A''$ and arrived at $\Phi^{+A''\tilde{A''}}$ in Eq. \eqref{Eq:Ea-cv-2}.

We next take Kraus-operator decompositions $\mc{E}_z(\cdot)=\sum_{k}N_{z,k}(\cdot)N_{z,k}^\dagger$ with each $N_{z,k}:\mbb{C}^{d_{A'}}\to\mbb{C}^{d_A}$.  Since $d_{A'}\geq d_A$, we can use the ``ricochet'' property $N_{z,k}\otimes\mbb{I}\ket{\phi_{d_{A'}}^+}^{A'\tilde{A'}}=\mbb{I}\otimes N_{z,k}^T\ket{\phi_{d_A}^+}^{A\tilde{A}}$ to obtain
\begin{align}
\label{Eq:Ea-cv-3}
\cv^*(\mbf{P}) = &\frac{1}{d_{A'}}\sum_z\sum_k\tr\left[\wt{P}_{z,k}\left(\mc{N}^{A\to B}\otimes\id^{\tilde{A}}[\phi_{d_A}^{+A\tilde{A}}]\right)\right]\notag\\
= & \tr[\Omega^{AB} J_\mc{N}^{AB}],
\end{align}
where $\wt{P}_{z,k} := (\mbb{I}^B\otimes N_{z,k}^*)\hat{\Pi}_z^{B\tilde{A'}}(\mbb{I}^B\otimes N_{z,k}^T)$, we have swapped the ordering of the systems to match earlier notation, and
\begin{align}
\Omega^{AB}&=\frac{1}{d_{A'}}\sum_z\sum_k(N_{z,k}^*\otimes\mbb{I}^B )\hat{\Pi}_z^{A'B}(N_{z,k}^T\otimes\mbb{I}^B)\notag\\
&=\frac{1}{d_{A'}}\sum_z\mc{E}^{*A'\to A}_z\otimes\id^B\left(\hat{\Pi}_z^{A'B}\right),
\end{align}
in which $\mc{E}^*_z(\cdot):=\sum_{k}N^*_{z,k}(\cdot)N_{z,k}^T$.  Since each $\mc{E}_z^*$ is trace-preserving, we have
\begin{align}
\tr_A\Omega^{AB}&=\frac{1}{d_{A'}}\sum_z\tr_A\left[\mc{E}^{*A'\to A}_z\otimes\id^B\left(\hat{\Pi}_z^{A'B}\right)\right]\notag\\
&=\frac{1}{d_{A'}}\sum_z\tr_{A'}\left(\hat{\Pi}_z^{A'B}\right)\notag\\
&=\frac{1}{d_{A'}}\tr_{A'}\left(\mbb{I}^{A'}\otimes\mbb{I}^B\right)=\mbb{I}^B.
\end{align}
Hence $\tr_A\Omega^{AB}=\mbb{I}^B$ is a necessary condition on the operator $\Omega^{AB}$ such that $\cv^*(\mbf{P})=\tr[\Omega^{AB}J_{\mc{N}}^{AB}]$.  Let us now sure that it is also sufficient.  

Consider any positive operator $\Omega^{AB}$ such that $\tr_A\Omega^{AB}=\mbb{I}^B$.  Introduce the generalized Pauli operators on system $A$, explicitly given by $U_{m,n}=\sum_{k=0}^{d_A-1}e^{i mk 2\pi/d_A}\op{m\oplus n}{m}$, where $m,n=0,\dots,d_A-1$ and addition is taken modulo $d_A$.  It is easy to see that $\Delta(\cdot):=\frac{1}{d_A}\sum_{m,n}U_{m,n}(\cdot)U_{m,n}^\dagger$ is a completely depolarizing map; i.e. $\Omega(X)=\tr[X]\mbb{I}$.  Hence, 
\[\Delta^A\otimes\id^B[\Omega^{AB}]=\mbb{I}^A\otimes \tr_A\Omega^{AB}=\mbb{I}^A\otimes\mbb{I}^B.\]
This implies that the elements $\{\mc{U}^A_{m,n}\otimes\id^B(\Omega^{AB})\}_{m,n}$ form a valid POVM on $AB$.  Therefore, we can construct an entanglement-assisted protocol as follows.  Let Alice and Bob share a maximally entangled state $\ket{\Phi^+_{d_A}}^{\tilde{A}A}$.  Alice applies the unitary encoding map on system $A$ given by $\mc{U}_{m,n}^T(\cdot):=U_{m,n}^T(\cdot) U_{m,n}^*$, and sends her system through the channel $\mc{N}$.  When Bob performs the POVM just described on systems $\tilde{A}B$, the obtained score is
\begin{align}
& \sum_{m,n}P(m,n|m,n) \notag \\
=& \frac{1}{d_A}\sum_{m,n}\tr\bigg[\left(\mc{U}^{\tilde{A}}_{m,n}\otimes\id^B\left[\Omega^{\tilde{A}B}\right]\right)(\id^{\tilde{A}}\otimes\mc{N}) \notag \\ 
& \hspace{3.75cm} \left(\id^{\tilde{A}}\otimes\mc{U}^{TA}_{m,n}\left[\Phi^{+ \tilde{A}A}_{d_A}\right]\right)\bigg]\notag\\
=& \frac{1}{(d_A)^2}\sum_{m,n}\tr\left[\Omega^{\tilde{A}B}\left(\id^{\tilde{A}}\otimes\mc{N}\left[\phi_{d_A}^{+\tilde{A}A}\right]\right)\right]\notag\\
=& \tr[\Omega J_{\mc{N}}].
\end{align}
The key idea in this equation is that the unitary encoding $U_{m,n}$ performed on Alice's side is canceled by exactly one POVM element on Bob's side.  This completes the proof of Theorem \ref{Thm:ea-cv}.
\end{proof}

\begin{remark}
The achievability protocol in the previous proof is essentially the original superdense coding protocol applied on a $d_A$-dimensional input channel.
\end{remark}

Theorem \ref{Thm:ea-cv} shows that the ea cv can be computed using semi-definite programming.  Here we provide a family of channels in which it can be computed even easier.
\begin{corollary}
Let $\mc{N}\in\cptp(A\to B)$ be any channel such that $J_\mc{N}$ has an eigenvector $\ket{\varphi}^{AB}$ with largest eigenvalue $\lambda_{\max}(J_{\mc{N}})$ such that $\varphi^B=\mbb{I}/d_B$.  Then
\begin{equation}
\cv^*(\mc{N})=d_B\lambda_{\max}(J_{\mc{N}}).
\end{equation}
\end{corollary}
\begin{proof}
Choose $\sigma^{AB}=d_B\op{\varphi}{\varphi}^{AB}$ in Theorem \ref{Thm:ea-cv}.  Clearly this choice is optimal.
\end{proof}
In addition, a solution can easily be deduced for all qubit channels.
\begin{theorem}
\label{Thm:qubit-ea-cv}
For a qubit channel $\mc{N}$, let $A$ be the $3\times 3$ correlation matrix of $J_\mc{N}$; i.e. $A_{ij}=\frac{1}{2}\tr[(\sigma_i\otimes\sigma_j) J_\mc{N}]$.  Then
\begin{equation}
\cv^*(\mc{N})=1+\Vert A\Vert_1
\end{equation}
where $\Vert A\Vert_1=\tr\sqrt{A^\dagger A}$.
\end{theorem}
\begin{proof}
Using Theorem \ref{Thm:ea-cv}, we can write $\Omega=\frac{1}{2}\left((\mbb{I}+\mbf{r}\cdot\vec{\sigma})\otimes\mbb{I}+\sum_{i,j}t_{ij}\sigma_i\otimes\sigma_j\right)$.  On the other hand, up to local unitaries, the Choi matrix of a channel $\mc{N}$ can be expressed as $J_\mc{N}=\frac{1}{2}(\mbb{I}\otimes(\mbb{I}+\mbf{s}\cdot\vec{\sigma})+\allowbreak\sum_{i}a_{i}\sigma_i\otimes\sigma_i)$.  Hence
\begin{align}
\tr[\Omega J_\mc{N}]=1+\sum_{i=1}^3a_i t_{ii}\leq 1+\sum_{i=1}^3|a_i|,
\end{align}
where the last inequality follows form the fact that $|t_{ii}|\leq 1$ since $\Omega\geq 0$.  The theorem is proven by recalling that the $|a_i|$ are the singular values of the correlation matrix $A$.
\end{proof}

It is worthwhile to compare Theorems \ref{Thm:qubit-cv} and \ref{Thm:qubit-ea-cv}.  Since $\Vert A\Vert_1\leq 3\sigma_{\max}(A)$ and $\sigma_{\max}(A)\leq 1$, we have
\begin{align}
\cv^*(\mc{N})=1+\Vert A\Vert_1\leq 2+2\sigma_{\max}(A)=2 \cv(\mc{N}).
\end{align}
Hence, the shared entanglement between sender and receiver cannot offer a multiplicative enhancement in the cv larger than the dimension.  In general, we conjecture the following.

\noindent\textit{Conjecture}:  For any channel $\mc{N}\in\cptp(A\to B)$,
\begin{equation}
\cv^*(\mc{N})\leq d_B\cdot \cv(\mc{N}).
\end{equation} 

\section{Relaxations on the Communication Value}\label{sec:relaxations_of_cv}
In previous sections we have made use of the fact the communication value can be expressed as a conic optimization problem (Proposition \ref{Prop:Proposition-cv}). It was noted in generality this problem would be hard to solve, but if the dimension the Choi matrix was sufficiently small, we could relax the cone, $\sep(A:B)$, to the PPT cone, $\ppt(A:B)$, and still determine $\cv(\mN)$ (Corollary \ref{cor:Proposition-PPT}). Moreover, in Section \ref{sec:entropic-characterization-of-cv}, we used the optimization program of $H_{\min}$ to justify characterizing $\cv(\mN)$ by a restricted min-entropy, and in Section \ref{sec:entanglement_assisted_cv} we saw the relationship between $H_{\min}$ and $\cv^{*}$. In all of these cases, we have considered the same optimization program and simply varied the cone to which the variable was restricted. That is, we have considered the general conic program 
\begin{equation}
      \begin{aligned}\label{eqn:conicPrimal}
        \text{maximize:}\quad & \tr[X \Omega^{AB}] \\
        \text{subject to:}\quad & \tr_{A}(X) = \mbb{I}^{B} \\
        & \Omega^{AB} \in \mK 
      \end{aligned}
\end{equation}
where $\cv(\mN)$ corresponds to $\mK=\sep(A:B)$ and $\cv^{*}$ corresponds to $\mK = \mathrm{Pos}(A \otimes B)$. It follows whenever we pick a cone $\mK$ such that $\sep(A:B) \subset \mK$, we obtain an upper bound on $\cv(\mN)$. Throughout the rest of this section, when considering relaxation $\sep(A:B) \subset \mK$, we denote the value of the optimization program by $\cv^{\mK}(\mN)$. In this section we primarily consider the PPT relaxation, $\mK=\ppt(A:B)$. We also discuss the relaxation to the $k$-symmetric cone, which is known to converge to the separable cone as $k$ goes to infinity \cite{Doherty-2004}, making it particularly relevant. 

\subsection{Multiplicativity of Tensored PPT Operators over the PPT cone}
We begin with the relaxation to the PPT cone. The primary advantage of this relaxation is that the problem becomes a semidefinite program and so pre-existing software may be used to find the optimal value. One may derive the primal and dual problems to be: 
\begin{center}
      \emph{Primal problem}\\[-5mm]
      \begin{equation}
      \begin{aligned}\label{eqn:PPTPrimal}
            \text{maximize:}\quad & \tr[X\Omega^{AB}]\\
            \text{subject to:}\quad & \tr_{A}(\Omega^{AB}) = \mbb{I}^{B} \\
            & \Gamma(\Omega^{AB}) \geq 0 \\
            & \Omega^{AB} \geq 0 \ .
      \end{aligned}
      \end{equation}
    \\
      \emph{Dual problem}\\[-5mm]
      \begin{equation}
      \begin{aligned}\label{eqn:PPTDual}
            \text{minimize:}\quad & \tr(Y_{1}) \\
            & \mbb{I}^{A} \otimes Y_{1} - \Gamma^{B}(Y_{2}) \geq \sigma \\
            & Y_{2} \geq 0 \\
            & Y_{1} \in \mathrm{Herm}(B) \ ,
      \end{aligned}
      \end{equation}
\end{center}
where $\Gamma^{B}$ is the partial transpose map on the $B$ space. This SDP satisfies strong duality as can be verified using Slater's condition.

With this established, we will now present a special multiplicativity property of the PPT relaxation, $\cv^{\ppt}$.
\begin{theorem}\label{thm:ppt-multiplicativity}
Let $R \in \ppt(A_{1}:B_{1}), Q \in \ppt(A_{2}:B_{2})$. Then 
$$\cv^{\ppt}(R\otimes Q) = \cv^{\ppt}(R) \, \cv^{\ppt}(Q) \ . $$
\end{theorem}
\begin{proof}
Let $R \in \ppt(A_{1}:B_{1}), Q \in \ppt(A_{2}:B_{2})$. Let $(Y_{1},Y_{2}),(\ol{Y}_{1},\ol{Y}_{2})$ be the dual optimizers for $R,Q$ respectively. From \eqref{eqn:PPTDual}, we have
\begin{equation}\label{eqn:ppt_mult_deriv_1}
    \begin{aligned}
        \mbb{I}^{A_{1}} \otimes Y_{1} \geq R + \Gamma^{B_{1}}(Y_{2}) \\ \mbb{I}^{A_{2}} \otimes \ol{Y}_{1} \geq Q + \Gamma^{B_{2}}(\ol{Y}_{2}) \ .
    \end{aligned}
\end{equation}
Define $R' := \Gamma^{B_{1}}(R), \, Q' := \Gamma^{B_{2}}(Q)$, which are both positive operators by assumption. Then we have
\begin{align*}
    & (\mbb{I}^{A_{1}} \otimes Y_{1}) \otimes (\mbb{I}^{A_{2}} \otimes \ol{Y}_{1}) \\
    \geq &(R + \Gamma^{B_{1}}(Y_{2})) \otimes (Q + \Gamma^{B_{2}}(\ol{Y}_{2})) \\
    = &R \otimes Q + R \otimes \Gamma^{B_{2}}(\ol{Y}_{2})  \notag\\
    &\quad +\Gamma^{B_{1}}(Y_{2}) \otimes Q + \Gamma^{B_{1}}(Y_{2}) \otimes \Gamma^{B_{2}}(\ol{Y}_{2}) \\
    = &R \otimes Q + \Gamma^{B_{1}B_{2}}(R' \otimes \ol{Y}_{2})\notag\\
    &\quad+ \Gamma^{B_{1}B_{2}}(Y_{2} \otimes Q') + \Gamma^{B_{1}B_{2}}(Y_{2} \otimes \ol{Y}_{2}) \\
    = &R \otimes Q + \Gamma^{B_{1}B_{2}}(R' \otimes \ol{Y}_{2} + Y_{2} \otimes Q' + Y_{2} \otimes \ol{Y}_{2}) \ ,
\end{align*}
where the first line follows from \eqref{eqn:ppt_mult_deriv_1}, the third is because of how the partial transpose over multiple systems may be decomposed, and the fourth is by linearity. Note that $R',Q'$ are positive as $R,Q$ are PPT. Moreover $Y_{2},\ol{Y}_{2}$ are positive by \eqref{eqn:PPTDual}. Thus the whole argument of $\Gamma^{B_{1}B_{2}}$ is a positive semidefinite operator. Therefore $(Y_{1,new} = Y_{1} \otimes \ol{Y}_{1}, Y_{2,new} = R' \otimes \ol{Y}_{2} + Y_{2} \otimes Q' + Y_{2} \otimes \ol{Y}_{2})$ is a feasible point of the dual problem for $R \otimes Q$, and it achieves an optimal value of $\tr(Y_{1})\tr(Y_{2})$. If we let $X_{1},X_{2}$ be the optimizers for the primal problem for $R,Q$ respectively, then $X_{1} \otimes X_{2}$ is clearly a feasible point for the primal problem for $R \otimes Q$ that achieves optimal value $\tr(RX_{1})\tr(QX_{2})$. Using the strong duality of the program, we have $\tr(RX_{1})\tr(QX_{2}) = \tr(Y_{1})\tr(Y_{2})$, so by strong duality we know our proposed optimizers are optimal and this completes the proof.
\end{proof}
\begin{corollary}\label{corr:ppt-multiplicativity}
Given any two co-positive maps, $\mN, \mM$, $\cv^{\ppt}(\mN \otimes \mM) = \cv^{\ppt}(\mN) \cv^{\ppt}(\mM)$.
\end{corollary}
It is interesting to note that we do not know that if only one of the channels is co-positive, then $\cv^{\ppt}$ is multiplicative, which would be a stronger claim. This is relevant because we conjecture that $\cv(\mN \otimes \mM)$ is multiplicative if either of the channels is entanglement breaking, which is known to hold for maximal $p$-norms for $p \geq 1$ \cite{King-2002}, but even the weaker case of multiplicativity where both channels are entanglement-breaking remains open and would mirror Corollary \ref{corr:ppt-multiplicativity}, but for separable Choi matrixs and the separable cone optimization.

\subsubsection*{Relation to $k$-Symmetric Extendable Cone}
Given the multiplicativity of tensors of PPT operators for $\cv^{\ppt}$, one might hope this property might hope this property extends to $\cv^{\sym_{k}}$ with $k$-symmetrically extendable operators, where an operator $R \in \Pos(A \otimes B)$ is $k$-symmetrically extendable if there exists $\tilde{R}^{AB^{k}_{1}} \in \Pos(A \otimes B^{\otimes k})$ such that
\begin{enumerate}
    \item $\tilde{R} = (\mbb{I}_{A} \otimes W_{\pi})\tilde{R}(\mbb{I}_{A} \otimes W_{\pi})^{\ast}$ for all $\pi \in \mathcal{S}_{k}$
    \item $\tr_{B^{k}_{2}}(\tilde{R}) = R$
    \item $(\mbb{I}_{A} \otimes T^{B} \otimes \mbb{I}_{\overline{B}^{k}_{2}})\tilde{R} \geq 0$ \ .
\end{enumerate}
Note that the $k$-symmetric extendable operators form a cone defined by semidefinite constraints. Moreover, it is known $\underset{k \to \infty}{\lim} \sym_{k} = \sep(A:B)$ \cite{Doherty-2004}. One can then attempt to extend Theorem \ref{thm:ppt-multiplicativity} in this setting. One can do this by deriving the dual program for $\cv^{\sym_{k}}$:
\begin{equation}\label{eqn:SimplekSymDual}
\begin{aligned}
    \text{min:}\quad & \tr(W) \\
    & \mbb{I}_{A} \otimes W \otimes \mbb{I}_{\overline{B}^{k}_{2}} + \sum_{j=1}^{k!} \left(Y_{j} - \Phi_{\pi^{-1}_{j-1}}(Y_{j})\right) \\
    & \hspace{3cm} \succeq \sigma \otimes \mbb{I}_{B_{2}^{k}} + \Gamma^{B_{1}}(Z)\\
    & Y_{j} \in \mathrm{Herm}(AB_{1}^{k}) \quad \forall j \in [k!]\\ 
    & Z \geq 0 \\
    & W \in \text{Herm}(B) \ ,
\end{aligned}
\end{equation}
where the indexing of $\pi_{j}$ is given by a chosen bijection between the index set $[k!]$ and the permutations in $\mathcal{S}_{k}$.
However, the proof method for Theorem \ref{thm:ppt-multiplicativity} does not seem to naturally extend due to the permutations of the spaces.

\subsection{Numerical Evaluation of the Communication Value}

\label{Sect:Numerics}

To numerically support this work, we developed the
CVChannel.jl software package which is publicly available on Github \cite{CVChannel2021}.
This Julia \cite{bezanson2017julia} software package provides tools for bounding the communication value of quantum channels and certifying their non-multiplicativity.
Our software is built upon the disciplined convex programming package, Convex.jl \cite{convexjl2014}, and our numerical results are produced using the splitting conic solver (SCS) \cite{scs2019}.
For more details, the curious reader should review the software documentation and source code found on our Github repository \cite{CVChannel2021}.

The communication value is difficult to compute in general, but it can be bounded with relative efficiently.
CVChannel.jl provides the following methods for bounding $\cv(\N)$.
An upper bound on $\cv(\N)$ is computed via the dual formulation of the PPT relaxation of the communication value Eq. \eqref{eqn:PPTDual},
\begin{equation}
    \cv(\N)\leq \cv^{\ppt}(\N).
\end{equation}
While $\cv^{\ppt}(\N)$ is a natural upper bound of $\cv(\N)$, we consider the dual specifically so that we take a conservative approach to numerical error.
That is, numerical error in minimizing the dual will result in a looser upper bound.  In general, when considering upper bounds we work with the dual problem and when considering lower bounds we work with the primal problem.  While the SDP satisfies strong duality, this guarantees minimizing false positives

For a lower bound on $\cv(\mc{N})$, we take a biconvex optimization approach to the problem
\begin{equation}
\cv(\mc{N})=\max_{\{\Pi_x\}, \{\rho_x\}}\sum_{x=1}^{d_B^2}\tr[\Pi_x\mc{N}(\rho_x)].
\end{equation}
This ``see-saw" technique is applied to similar problems in \cite{Reimpell2005,Kosut2009}, although our implementation remains distinct.
To begin, an ensemble of pure quantum states $\{\rho_x\}_{x=1}^{d_B^2}$ are initialized at random according to the Haar measure.
Then, the following procedure is iterated:
\begin{enumerate}
    \item With the states fixed, the POVM measurement is numerically optimized as a semidefinite program
    \begin{equation}
        \max_{\{\Pi_y\}_{y=1}^{d_B^2}}\sum_{x=1}^{d_B^2}\tr[\Pi_x\mc{N}(\rho_x)].
    \end{equation}
    \item With optimal measurement as $\{\Pi_y^{\star}\}$, we compute the optimal ensemble of quantum states $\{\rho_x^{\star}\}_{x=1}^{d_b^2}$ as
    \begin{equation}
        \rho_x^{\star} = ||\mc{N}^{\dagger}(\Pi_x)||_{\infty},
    \end{equation}
    where $\mc{N}^{\dagger}$ is the adjoint channel and $||\cdot||_{\infty}$ denotes the largest eigenvalue.
\end{enumerate}
Repeating this procedure results in a set of optimized states $\{\rho_x^{\star}\}_{y=1}^{d_B^2}$ and measurement $\{\Pi_y^{\star}\}_{y=1}^{d_B^2}$ such that
\begin{equation}
    \cv^{SeeSaw}(\N)=\sum_{x=1}^{d_B^2}\tr[\Pi_x^{\star}\mc{N}(\rho_x^{\star})] \leq \cv(\N).
\end{equation}
To improve the see-saw optimization, the procedure is simply performed many times with randomly initialized states.
Combining these techniques, we numerically bound the communication value, 
\begin{equation}
    \cv^{SeeSaw}(\N) \leq \cv(\N)\leq \cv^{\ppt}(\N, \;\text{dual}).
\end{equation}

To numerically certify that quantum channels $\N$ and $\mc{M}$ are non-multiplicative, we need to compute a lower bound on $\cv(\N\otimes\mc{M})$ and upper bound on $\cv(\N)$ and $\cv(\mc{N})$.
CVChannel.jl computes the lower bound as $\cv^{SeeSaw}(\N\otimes\mc{M})\leq \cv(\N\otimes\mc{M})$ and the upper bound as $\cv(\N)\leq \cv^{\ppt}(\N, \; \text{dual})$.
Non-multiplicativity is numerically confirmed when
\begin{equation}
    \cv^{SeeSaw}(\N\otimes\mc{M}) - \cv^{\ppt}(\N)\cv^{\ppt}(\mc{M}) > \varepsilon
\end{equation}
where $\varepsilon>0$ is a conservative bound to the numerical error.
One drawback of this procedure is its susceptibility to false negatives due to the fact that the PPT Relaxation is a loose upper bound of the communication value.

\subsection{Examples}
Having established properties of the $\ppt$ relaxation of the communication value, we investigate channels which are known in other settings to admit non-multiplicative behaviour. In particular, we look at the family of Werner-Holevo channels \cite{Werner-2002a}, the dephrasure channel \cite{Leditzky-2018}, and the Siddhu channel \cite{Siddhu-2020}. We see that the Werner-Holevo channel is not multiplicative over a range of parameters, but the dephrasure and Siddhu channel which are known for their superactivation of coherent information are always multiplicative for the communication value. In some sense this should not be surprising as communication value captures a notion of using the quantum channel to transmit classical information whereas the coherent information measures the ability to transfer quantum information. However, it exemplifies how different the coherent information and communication values are as measures.
\subsubsection*{Werner-Holevo Channels}
In Section \ref{Sect:Werner-Holevo}, we showed how to determine the $\cv$ for the Werner-Holevo channels. In this section we extend the method to obtain this result to the construction of a linear program (LP) for determining the $\cv^{\ppt}$ for $n$ Werner-Holevo channels ran in parallel. We then use this to show non-multiplicativity for $\cv(\mc{W}_{d,\lambda} \otimes \mc{W}_{d,\lambda})$ as a function of $\lambda$, as well as the non-multiplicativity of $\cv^{\ppt}$ for more copies of the channel. We note our derivation assumes the dimension is the same for all channels, but a generalization is straightforward.
\begin{proposition}
Considering $n$ Werner-Holevo channels, there is a linear program 
$$ \max\{\langle a , c \rangle \, : \, Ac \geq 0 , \, Bc \geq 0, \, \langle g , c \rangle = 1 \} \ , $$
which obtains the value of $\cv^{\ppt}(\otimes_{i=1}^{n} J(\mc{W}_{d,\lambda_{i}}))$. Moreover, there exists an algorithm to generate the constraints $a,A,B,g$ which takes at most $\mathcal{O}(n2^{2n})$ steps.
\end{proposition}
\begin{proof}[Derivation of Constraints]
Let $\Pi_{0} := \Pi_+$, $\Pi_{1} := \Pi_-$. This labelling will simplify notation.
We are interested in $\cv^{\ppt}(\bigotimes_{i=1}^{n} J(\mc{W}_{d,\lambda_{i}}))$.
Recalling the objective function of \eqref{eqn:PPTPrimal} is 
$$\tr[\bigotimes_{i=1}^{n} J(\mc{W}_{d,\lambda_{i}}) \Omega^{A^{n}B^{n}}] \ ,$$
we can twirl $\Omega$ by moving the symmetry of the Holevo channels onto $\Omega$. This results in $\Omega = \sum_{s \in \{0,1\}^{n}} c_{s} R_{s}$ where 
$$ R_{s} = \bigotimes_{i=1}^{n} \mbb{F}^{s(i)} = \sum_{j \in \{0,1\}^{n}} \left( \bigotimes_{i \in [n]} (-1)^{s(i)\wedge j(i)} \Pi_{j(i)} \right) \ , $$
where the constraint on the sign is because $\mbb{F}^0 = \mbb{I} = \Pi_0 + \Pi_1$ and $\mbb{F}^1 = \Pi_0 - \Pi_1$, so a term is negative iff $s(i)=j(i)=1$. Combining these, we can express $\Omega$ as a linear combination of orthogonal subspaces with coefficients stored in a vector $c$: 
\begin{align}\label{eq:general-wh-lp-optimizer-decomp}
    \Omega = \sum_{s \in \{0,1\}^{n}} c_{s} \sum_{j \in \{0,1\}^{n}} \left( \bigotimes_{i \in [n]} (-1)^{s(i)\wedge j(i)} \Pi_{j(i)} \right) \ .
\end{align}
With the state simplified into mutually orthogonal subspaces, we just need to convert the constraints of \eqref{eqn:PPTPrimal} to constraints on $c \in \mbb{R}^{2^{n}}$. 

Guaranteeing positivity of $\Omega$ is equivalent to guaranteeing the weight of each orthogonal subspace in \eqref{eq:general-wh-lp-optimizer-decomp} is non-negative. As multiple elements of $c$ can have weight on multiple subspaces, the constraint is that the relevant linear combination of $c$ is non-negative for each subspace. Thus the positivity constraints may be written as $Ac \geq 0$ where $A \in \mbb{R}^{2^{n} \times 2^{n}}$ matrix storing the sign information $(-1)^{s(i)\wedge j(i)}$ for all $s$,$j$.

The PPT constraints correspond to $\Omega^{\Gamma} \geq 0$. Noting that $\mbb{F}^{\Gamma} = \Phi^{+}$, the unnormalized maximally entangled state. We have
$ \Omega^{\Gamma} = \sum_{s \in \{0,1\}^{n}} c_{s} \bigotimes_{i \in [n]} X^{s(i)} \ , $
where
$$ X^{s(i)} := \begin{cases} d^{-1}(\Phi^{\perp}+\Phi^{+}) & s(i) = 0 \\ \Phi^{+} &  s(i) = 1 \end{cases} \ , $$
where $\Phi^{\perp} = d\mbb{I} - \Phi^{+}$. In other words, we have decomposed $\Omega^{\Gamma}$ into linear combinations of a set of orthogonal subspaces.\footnote{We note this implies $d^{-1}(\Phi^{\perp}+\Phi^{+}) = \mbb{I}$. The choice of presentation is to make it clear we are considering two orthogonal subspaces.} Again, we only need to store the constraints on $c$ which in this case is the order of $d$ and if the coefficient is zero. By the definition of $X^{s(i)}$, there is not weight of a subspace for $c_{s}$ iff the $i^{th}$ element in the tensor is $\Phi^{\perp}$ and $s(i) = 1$, and otherwise the weight is given by $d^{-(2^{n}-w(s))}$ where $w(\cdot)$ is the Hamming weight of the string $s$. Thus the PPT constraints may be written as $Bc \geq 0$ where $B \in \mbb{R}^{2^{n} \times 2^{n}}$.

Recalling $\Omega = \sum_{s \in \{0,1\}^{n} c_{s} R_{s}}$ and $\tr_{A}(F) = \mbb{I}^{B}$, $\tr_{A}(\mbb{I}) = d\mbb{I}^{B}$, the partial trace condition is reduced to $\langle g , c \rangle = 1$ where $g \in \mathbb{R}^{2^{n}}$ and $g(s) = d^{n-w(s)}$.

Finally, we have the objective function. We write $\bigotimes_{i=1}^{n} J(\mc{W}_{d,\lambda_{i}}) = \sum_{s \in \{0,1\}^{n}} \left(\bigotimes_{i} \zeta_{i}(s(i))\Pi_{s({i})} \right)$, where $\zeta_{i}(s(i)) = \begin{cases} \lambda_{i} f_0 & i = 0 \\ (1-\lambda_{i})f_{1} & i = 1 \end{cases}$ where $f_{i}$ is the normalization constant in front of the projector. Calculating $\tr[\bigotimes_{i=1}^{n} J(\mc{W}_{d,\lambda_{i}})\Omega]$ using the above expression along with \eqref{eq:general-wh-lp-optimizer-decomp}, one can simplify the objective function to
$$ \sum_{s \in \{0,1\}^{n}} c_{s} \left( \sum_{j \in \{0,1\}^{n}}\left[ \prod_{i \in [n]} (-1)^{s(i)\wedge j(i)} \varphi_{i}(j(i)) \right] \right) \ , $$
where $\varphi_{i}(j(i))$ is the same as $\zeta_{i}(j(i))$, except without the normalization constant. Thus we may define $a$ as the argument of the large parentheses. This completes the derivation of the LP. Finally we note to construct the constraints one needs to loop through nested loops of sizes $2^{n},2^{n},n$ which results in the $\mathcal{O}(n2^{2^n})$ steps in the algorithm.
\end{proof}
Using these numerics, we can look at the behaviour of the PPT-relaxation of the communication value of the $n$-fold Werner Holevo Channel (Fig. \ref{fig:multi-copy-werner-holevo}). We can see that the non-multiplicativity over the PPT cone grows exponentially (Fig. \ref{fig:multi-copy-werner-holevo}) and that all non-multiplicativity dies out at $\lambda=0.3$ in all cases. We note it is known that for the tensor product of two Werner states, the space of PPT operators is the same as the space of separable operators. In this case, we see the non-additivity of the true communication value for the Werner-Holevo channels.
\begin{figure}[h]
    \centering
    \includegraphics[width = \columnwidth]{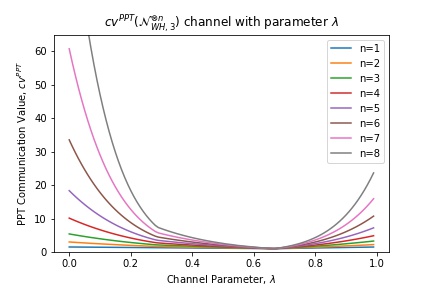}
    \caption{The $\cv^{\ppt}$ value of the $n$-fold Werner Holevo channel for all values of channel parameter $\lambda \in [0,1]$.}
    \label{fig:multi-copy-werner-holevo}
\end{figure}
\begin{figure}[h]
    \centering
    \includegraphics[width = \columnwidth]{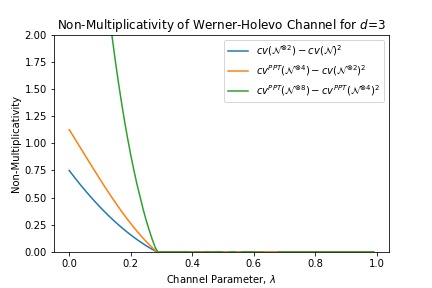}
    \caption{Here we see the non-multiplicativity of tensoring the Werner-Holevo channel with itself. We note that the blue line characterizes the multiplicativity of $\cv$ rather than just $\cv^{\ppt}$.}
    \label{fig:non-multiplicativity}
\end{figure}

\subsubsection{PPT Relaxation of Werner-Holevo with the Identity}\label{subsec:WHID}
An immediate corollary of Theorem \ref{thm:covariance-multiplicativity} is that the Werner-Holevo channel when ran in parallel with the identity channel of any dimension is multiplicative. That is, $\cv(\mc{W}_{d,\lambda} \otimes id_{d'}) = d' \cdot\cv(\mc{W}_{d,\lambda})$. However, here we find that this is not the case for $\cv^{PPT}$ which is non-multiplicative, exhibiting a clear separation between the $\cv$ and its relaxation. This separation is given for the $\mc{W}_{d,0} \otimes id_{d'}$ in Fig. \ref{fig:cv-ppt-versus-cv-wh-with-id}. It is determined using the following proposition.
\begin{proposition}
The PPT communication value of the Werner-Holevo channel ran in parallel with an identity channel, $\cv^{\ppt}(\mc{W}_{d,\lambda} \otimes\id_{d'})$, is given by the linear program
\begin{equation}
\begin{aligned}
\max & \;\; dd'[w+yd'+(2\lambda-1)(x+zd')]\\
&\;\; 0\leq w-x+d'y-d'z\\
&\;\; 0\leq w-x\\
&\;\; 0\leq w+x+d'y+d'z\\
&\;\; 0\leq w+x&\\
&\;\; 0\leq w+dx-y-dz\\
&\;\; 0\leq w-y \\
&\;\; 0\leq w+dx+y+dz \\
&\;\; 0\leq w+y\\
&\;\; 1=dd'w+d'x+dy+z.
\end{aligned}
\end{equation}
\end{proposition}
\begin{proof}[Derivation]
The derivation is similar to that of the previous $\cv^{PPT}$ LP derivations, we just also consider $\ol{V}V$ covariance for the identity channel. Let us consider the channel $\mc{W}_{d,\lambda} \otimes\id_{d'}$, where $\mc{W}_{d,\lambda}$ is defined in \eqref{eqn:WH-defn}. Then $\mc{J}_{\mc{W}}\otimes \phi^+_{d'}$ is $UU\ol{V}V$-covariant, and so for any feasible operator $\sigma^{AB}$ in the $\cv^{PPT}$ SDP, we have
\begin{align}
    & \tr[\sigma^{AA:BB'}J_{\mc{W}}\otimes J_{\id_{d'}}] \notag \\ = & \tr[\sigma^{AA':BB'}\mc{T}_{UU}(\mc{J}_{\mc{W}})\otimes\mc{T}_{\ol{V}V}(\phi^+_{d'})]\notag\\
    =& \tr[\mc{T}_{UU}\otimes\mc{T}_{\ol{V}V}(\sigma^{AA':BB'})J_{\mc{W}}\otimes \phi^+_{d'}].\notag
\end{align}
Note that $\mc{T}_{UU}\otimes\mc{T}_{\ol{V}V}(\sigma)$ is still a feasible operator, and so without loss of generality we can assume that $\sigma^{AA':BB'}$ is itself $UU\ol{V}V$-covariant.  Thus, we can parametrize $\sigma$ as
\begin{align*}
     w\mbb{I}^{AB}\otimes\mbb{I}^{A'B'}+x\mbb{F}_d\otimes\mbb{I}+y\mbb{I}\otimes\phi^+_{d'}+z\mbb{F}_d\otimes\phi^+_{d'}.
\end{align*}
The space of $UU\ol{V}V$ operators is spanned by the set of four orthogonal operators
\begin{equation*}
    \left\{
    \begin{array}{cc}
        \Pi_d^-\otimes\phi_{d'}^+ & \Pi_d^-\otimes(d'\mbb{I}-\phi_{d'}^+) \\[2mm]
        \Pi_d^+\otimes\phi_{d'}^+ & \Pi_d^+\otimes(d'\mbb{I}-\phi_{d'}^+)
    \end{array}
    \right\}
\end{equation*}
Positivity then amounts to the conditions
\begin{equation}\label{eqn:WHID-Pos-cond}
    \begin{aligned}
        w-x+d'y-d'z&\geq 0 \\
        w-x&\geq 0 \\
        w+x+d'y+d'z&\geq 0 \\
        w+x&\geq 0.
    \end{aligned}
\end{equation}
The partial transpose of $\sigma$, $\sigma^{\Gamma_{BB'}}$ is given by
\begin{align*}
    w\mbb{I}^{AB}\otimes\mbb{I}^{A'B'}+x\phi^+_d\otimes\mbb{I}+y\mbb{I}\otimes\mbb{F}_{d'}+z\phi^+_d\otimes\mbb{F}_{d'}.
\end{align*}
To check positivity, we now use just need to swap the orthogonal basis operators:
\begin{equation*}
    \left\{
    \begin{array}{cc}
    \phi_{d}^+\otimes\Pi_{d'}^- & (d\mbb{I}-\phi_{d}^+)\otimes\Pi_{d'}^- \\[2mm] \phi_{d}^+\otimes\Pi_{d'}^+ & (d\mbb{I}-\phi_{d}^+)\otimes\Pi_{d'}^+
    \end{array}
    \right\}
\end{equation*}
This yields the conditions
\begin{equation}\label{eqn:WHID-PPT-cond}
    \begin{aligned}
    w+dx-y-dz&\geq 0 \\
    w-y&\geq 0 \\
    w+dx+y+dz&\geq 0 \\
    w+y&\geq 0.
    \end{aligned}
\end{equation}
Finally, we compute the objective function
\begin{align}
& \tr[\sigma^{AA':BB'}J_\mc{W}\otimes\phi_{d'}^+] \notag\\
=&d'(w\tr[J_{\mc{W}}]+x\tr[\mbb{F}J_{\mc{W}}])+{d'}^2(y\tr[J_{\mc{W}}]+z\tr[\mbb{F}J_{\mc{W}}])\notag\\
=&d'(wd+xd(2\lambda-1))+{d'}^2(yd+zd(2\lambda-1))\notag\\
=&dd'[w+yd'+(2\lambda-1)(x+zd')], \label{eqn:WHID-obj-func}
\end{align}
and the partial trace condition
\begin{align}\label{eqn:WHID-tr-cond}
    \tr_{AA'}[\sigma^{AA':BB'}]&=(dd' w+d'x+dy+z)\mbb{I}^{BB'}.
\end{align}
Combining \eqref{eqn:WHID-Pos-cond} -- \eqref{eqn:WHID-tr-cond} completes the derivation.
\end{proof}

\begin{figure}[h]
    \centering
    \includegraphics[width=\columnwidth]{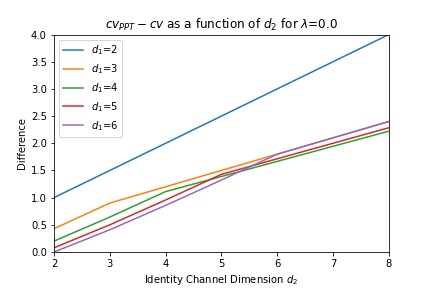}
    \caption{This shows the gap between $\cv(\mathcal{N}_{WH,d_{1},\lambda}\otimes \id_{d_{2}}) = d_{2} \cv(\mathcal{N}_{WH,d,\lambda})$ and $\cv^{PPT}(\mathcal{N}_{WH,d_{1},\lambda}\otimes \id_{d_{2}})$ for $\lambda = 0$.}
    \label{fig:cv-ppt-versus-cv-wh-with-id}
\end{figure}

\subsubsection*{Dephrasure Channel}
We next consider the dephrasure channel,
$$ \mathcal{N}_{p,q}(X) := (1-q)\left((1-p)\rho + pZ\rho Z \right) + q \tr(X) \dyad{e} \ , $$
where $p,q \in [0,1]$. The interesting aspect of the dephrasure channel is that in some parameter regime it admits superadditivity of coherent information \cite{Leditzky-2018}. We will first present it's communication value.  
\begin{lemma}
$\cv(\mathcal{N}_{p,q}) = 2-q$.
\end{lemma}
\begin{proof}
We are going to prove this by constructing feasible operators in the primal and dual which achieve this value. First we note the Choi matrix:
\begin{align*}
J(\mathcal{N}_{p,q}) =& (1-q)\left(\op{00}{00}+\op{11}{11}\right) \\ & + \gamma\left(\ket{00}\bra{11} + \ket{11}\bra{00}\right) \\ & \hspace{1cm} + q\left(\op{0e}{0e} + \dyad{1e}\right) \ ,
\end{align*}
where $\gamma := (1-q)(1-2p)$. Then for the primal problem, we may choose 
$$X =\dyad{00} + \dyad{11} + 1/2\left(\dyad{0e} + \dyad{1e}\right) \ . $$ This clearly satisfies $\tr_{A}(X) = \mbb{I}^{B}$, it is PPT as it is diagonal, and $\langle X , J(\mathcal{N}_{p,q}) \rangle = 2 -q$. For the dual problem, let 
\begin{align*}
Y_{1} &= (1-q)(\dyad{0}+\dyad{1}) + q\dyad{e} \\ 
Y_{2} &= \kappa (\dyad{01} + \dyad{10}) - \gamma(\ket{01}\bra{10}+\ket{10}\bra{01}) \ ,
\end{align*}
where $(1-q) \geq \kappa \geq |\gamma| = (1-q)|(1-2p)|$. Note this interval is never empty as $|(1-2p)|\in[0,1]$ for all $p \in [0,1]$.

Then $Y_{1}$ is clearly Hermitian, and $Y_{2} \succeq 0$ as it's eigenvalues are $\kappa \pm \gamma \geq |\gamma| \pm \gamma \geq 0$ and $0$ with multiplicity $4$. Then, one may calculate from these expressions that
\begin{align*}
& \mbb{I}_{A} \otimes Y_{1} - \Gamma(Y_{2}) - J(\mathcal{N}_{p,q}) \\
=& \left((1-q-\kappa \right) \left[\ket{00}\bra{11} + \ket{11}\bra{00}\right] \ ,
\end{align*}
Therefore we have constructed a feasible choice. Finally, $\tr(Y_{1}) = 2-q$ completes the proof.
\end{proof}
Note what the above implies is the `dephasing' property of the dephrasure is irrelevant. This is in some sense intuitive as the dephasing cannot hurt the classical information if the optimal strategy is sending data in the classical basis. Indeed, it is easy to see the above value may be achieved by using the signal states $\{\dyad{0},\dyad{1}\}$ and the projective measurement decoder $\{\dyad{0}+1/2\dyad{e},\dyad{1}+1/2\dyad{e}\}$ as then for both signal states you will guess correctly $(1-q)+q/2$ conditioned on the state sent. As one might expect, in such a situation the communication value of the channel would be multiplicative with itself. As we require an upper bound, we verify this by an exhaustive numerical search using the dual problem of $\cv^{\ppt}$.

\begin{theorem}
$\cv(\mathcal{N}_{p,q}^{\otimes 2}) =\cv(\mathcal{N}_{p,q})^{2}$, i.e. the dephrasure channel's communication value is multiplicative.
\end{theorem}
\begin{proof}
A search over the dual problem $\cv^{\ppt}(\mathcal{N}_{p,q}^{\otimes 2})$ for $p,q \in [0,0.01,...,1]$ is always within numerical error of $\cv(\mathcal{N}_{p,q})^{2}$. As the dual problem always obtains an upper bound on $\cv^{\ppt}$, and $\cv^{\ppt}$ is an upper bound on $\cv$, we may conclude that the dephrasure channel is multiplicative.
\end{proof}

\subsubsection*{Siddhu Channel}
Finally we consider the following family of channels:
\begin{align*}
    \mathcal{N}_{s}(X) := \sum_{i=0}^{1} K_{i}XK_{i}^{\dagger} \ ,
\end{align*}
where 
\begin{align*}
    K_{0} = \sqrt{s} \op{0}{0} + \op{2}{1} \quad K_{1} = \sqrt{1-s} \op{1}{0} + \op{2}{2} \ ,
\end{align*}
where $s \in [0,1/2]$. This channel is known to have non-additive coherent information over its entire parameter range when tensored with itself. However, we will now show the communication value of the channel is multiplicative with itself over the whole range.
\begin{lemma}
$\cv(\mathcal{N}_{s}) = 2$ for all $s \in [0,1/2]$.
\end{lemma}
\begin{proof}
Like the dephrasure channel, we prove this by constructing upper and lower bounds that are the same. \\

For a lower bound on $\cv(\mathcal{N}_{s})$, consider the encoding $\{\dyad{0},\dyad{1},\dyad{2}\}$ and the decoding $\{\dyad{0}+\dyad{1},\dyad{2}\}$. Note that for all $s \in [0,1/2]$, $\mathcal{N}_{s}(\dyad{0}) = s\dyad{0} + (1-s)\dyad{1}$ and $\mathcal{N}_{s}(\dyad{1}) = \mathcal{N}_{s}(\dyad{2}) = \dyad{2}$. Thus, with this encoding and decoding, we induce the conditional probability distribution $1 = \mathbf{P}(0|1) = \mathbf{P}(1|2) = \mathbf{P}(1|3)$ and zero otherwise. Thus we have $2 \leq \cv^{3 \to 2}(\mathcal{N}_{s}) \leq \cv(\mathcal{N}_{s})$.

For an upper bound, we consider the dual problem of $\cv^{\ppt}$ \eqref{eqn:PPTDual}. First note that we can write the Choi matrix as:
\begin{align*}
     J(\mathcal{N}_{s})
    = \begin{bmatrix}
        s & 0   & 0 & 0 & 0 & \sqrt{s} & 0 & 0 & 0 \\
        0 & 1-s & 0 & 0 & 0 & 0 & 0 & 0 & \sqrt{1-s} \\
        0 & 0 & 0 & 0 & 0 & 0 & 0 & 0 & 0 \\
        0 & 0 & 0 & 0 & 0 & 0 & 0 & 0 & 0 \\
        0 & 0 & 0 & 0 & 0 & 0 & 0 & 0 & 0 \\
        \sqrt{s} & 0 & 0 & 0 & 0 & 1 & 0 & 0 & 0 \\
        0 & 0 & 0 & 0 & 0 & 0 & 0 & 0 & 0 \\
        0 & 0 & 0 & 0 & 0 & 0 & 0 & 0 & 0 \\
        0 & \sqrt{1-s} & 0 & 0 & 0 & 0 & 0 & 0 & 1 \\
      \end{bmatrix} \ .
\end{align*}
Then we let $Y_{1} = s \dyad{0} + (1-s) \dyad{1} + \dyad{2}$ and 
\begin{align*}
    Y_{2} =
    \begin{bmatrix}
        0 & 0 & 0 & 0 & 0 & 0 & 0 & 0 & 0 \\
        0 & 0 & 0 & 0 & 0 & 0 & 0 & 0 & 0 \\
        0 & 0 & 1 & -\alpha & 0 & 0 & 0 & -\beta & 0 \\
        0 & 0 & -\alpha & s & 0 & 0 & 0 & \gamma & 0 \\
        0 & 0 & 0 & 0 & 0 & 0 & 0 & 0 & 0 \\
        0 & 0 & 0 & 0 & 0 & 1 & 0 & 0 & 0 \\
        0 & 0 & 0 & 0 & 0 & 0 & 0 & 0 & 0 \\
        0 & 0 & -\beta & \gamma & 0 & 0 & 0 & 1-s & 0 \\
        0 & 0 & 0 & 0 & 0 & 0 & 0 & 0 & 1 \\
      \end{bmatrix} \ ,
\end{align*}
where $\alpha = \sqrt{s}, \beta = \sqrt{1-s}, \gamma = \sqrt{s(1-s)}$, which is positive semidefinite as it has eigenvalues $2$ and $0$ with multiplicity eight. It is then easy to determine  
\begin{align*}
     & \mbb{I}_{A} \otimes Y_{1} - \Gamma(Y_{2}) - J(\mN_{s}) \\
   = &\begin{bmatrix}
        0 & 0 & 0 & 0 & 0 & \delta & 0 & 0 & 0 \\
        0 & 0 & 0 & 0 & 0 & 0 & 0 & 0 & \epsilon \\
        0 & 0 & 0 & 0 & 0 & 0 & 0 & 0 & 0 \\
        0 & 0 & 0 & 0 & 0 & 0 & 0 & 0 & 0 \\
        0 & 0 & 0 & 0 & 1-s & 0 & -\gamma & 0 & 0 \\
        \delta & 0 & 0 & 0 & 0 & 0 & 0 & 0 & 0 \\
        0 & 0 & 0 & 0 & -\gamma & 0 & s & 0 & 0 \\
        0 & 0 & 0 & 0 & 0 & 0 & 0 & 0 & 0 \\
        0 & \epsilon & 0 & 0 & 0 & 0 & 0 & 0 & 0 \\
      \end{bmatrix} \ ,
\end{align*}
where $\delta = \alpha - \sqrt{s}, \epsilon = \beta - \sqrt{1-s}$. One may verify that this has eigenvalues of $1$ and $0$ with multiplicity eight. Thus it is feasible and $\tr(Y_{1}) = 2$. Noting that $\cv(\mN_{s}) \leq \cv^{\ppt}(\mN_{s})$, we have $2 \leq\cv(\mN_{s}) \leq 2$, which completes the proof.
\end{proof}
\begin{theorem}
For all $s \in [0,1/2]$, $\cv(\mN_{s}^{\otimes 2}) =\cv(\mN_{s})^{2}$, i.e. the communication value is always multiplicative.
\end{theorem}
\begin{proof}
A numerical search of $\cv^{\ppt}(\mN_{s}^{\otimes 2})$ for $s \in [0,0.01,...,0.5]$ finds the value equals four within an error of $\leq 3 \times 10^{-6}$. As $\cv^{\ppt}$ upper bounds $\cv$, the channel is multiplicative.
\end{proof}

\section{Relationship to Capacities and No-Signalling}

\label{Sect:Capacity-NS}


As noted in the introduction, $\lceil\cv(\mN) \rceil$ captures the classical communication cost to perfectly simulate every classical channel induced by $\mN$ using non-signalling (NS) resources. This is because for a classical channel $\mathbf{P}$, the one-shot classical communication cost for zero-error simulation with classical NS, $\kappa_0^{\text{NS}}$, is given by $\lceil \sum_{y} \max_{x} p(y|x)\rceil$ \cite[Theorem 16]{Cubitt-2011a}. Noting that $\cv(\mathbf{P}) = \sum_{y} \max_{x} p(y|x)$, it follows $\kappa_0^{\text{NS}} =  \lceil \cv(\mathbf{P}) \rceil$. Furthermore, due to the multiplicativity of cv for classical channels, the no-signalling assisted zero-error simulation capacity is also given by $\kappa_{0}^{\text{NS}}$, as was remarked in the original paper. Moreover, it is easy to show the classical capacity of a classical channel is bounded by $\cv(\mathbf{P})$ \cite[Remark 17]{Cubitt-2011a}:
$$ C(\mathbf{P}) \leq \log\cv(\mathbf{P}) = \chi_{\max}(\mathbf{P}) \ , $$
where we have used Theorem \ref{Thm:chi-max-hmin} in the last equality. Losing the single-letter property, it is easy to generalize this to arbitrary quantum channels by using the Holevo-Schumacher-Westmoreland theorem \cite{Holevo-1998a,Schumacher-1997a},
\begin{align*} 
C(\mN) =& \underset{k\to \infty}{\lim} \frac{1}{k} \chi(\mN^{\otimes k}) \\ \leq& \underset{k\to \infty}{\lim} \frac{1}{k} \chi_{\max}(\mN^{\otimes k})\\ 
=& \underset{k\to \infty}{\lim} \frac{1}{k} \log(\cv(\mN^{\otimes k})) \ ,
\end{align*}
and whenever $\mN$ satisfies weak multiplicativity for $\cv$, such as for entanglement-breaking channels, this reduces to a single-letter upper bound.

In the entanglement-assisted regime, the relationships persist. First we recall that the SDP for min-entropy is multiplicative, and so $\mathcal{CV}^{*}(\mN) = \cv^{*}(\mN)$ for arbitrary quantum channel $\mN$. This aligns with the fact the entanglement-assisted capacity of a quantum channel, $C_{E}(\mN)$, is single-letter but the unassisted capacity is not. Continuing the parallels, $\lceil \cv^{*}(\mN) \rceil$ gives the classical communication cost to perfectly simulate $\mN$ with a quantum no-signalling resource \cite{Duan-2016a}. Given the above, a natural question is then if one can find bounds on the entanglement-assisted capacity, $C_{E}(\mN)$, in terms of $\cv^*(\mN)$. Indeed, this can be done by using the definition of $\cv^{*}$ and the fact that $\cv^{*}$ is characterized by minimal error discrimination (as in Eq. \eqref{eqn:cv-max-holevo-alt}),
$$ \cv^{*}(\mN) = \underset{\rho_{XAA'}}{\sup} |\mX| \exp(-H_{\min}(X|BC)_{(id_{X} \otimes \mN \otimes id_{A'})(\rho)}) \ , $$
where the supremum is over $\rho_{XAA'}$ such that $\rho_{X}$ is uniform and the state is homogenous on register $A'$ \cite{Holevo-2002a,Watrous-2018a}. It follows by the same manipulations used in Eq. \eqref{eqn:cv-max-holevo-alt} that 
\begin{equation}
    \log \cv^{*}(\mN) = \chi_{E,\max}(\mN),
\end{equation}
where $\chi_{E,\max}$ is the entanglement-assisted $\max$-Holevo information, which is straightforward to define using \cite{Holevo-2002a,Watrous-2018a,Beigi-2013a}. Since the entanglement-assisted capacity equals the regularized entanglement-assisted Holevo information, we can conclude the
$$ C_{E}(\mN) \leq \log \cv^{*}(\mN) \ , $$
where the regularization disappears because $\cv^*(\mN)$ is always multiplicative.

\bibliographystyle{IEEEtran}
\bibliography{cv-bibliography}

\begin{thebibliography}{10}
\providecommand{\url}[1]{#1}
\csname url@samestyle\endcsname
\providecommand{\newblock}{\relax}
\providecommand{\bibinfo}[2]{#2}
\providecommand{\BIBentrySTDinterwordspacing}{\spaceskip=0pt\relax}
\providecommand{\BIBentryALTinterwordstretchfactor}{4}
\providecommand{\BIBentryALTinterwordspacing}{\spaceskip=\fontdimen2\font plus
\BIBentryALTinterwordstretchfactor\fontdimen3\font minus
  \fontdimen4\font\relax}
\providecommand{\BIBforeignlanguage}[2]{{%
\expandafter\ifx\csname l@#1\endcsname\relax
\typeout{** WARNING: IEEEtran.bst: No hyphenation pattern has been}%
\typeout{** loaded for the language `#1'. Using the pattern for}%
\typeout{** the default language instead.}%
\else
\language=\csname l@#1\endcsname
\fi
#2}}
\providecommand{\BIBdecl}{\relax}
\BIBdecl

\bibitem{Cubitt-2011a}
T.~S. Cubitt, D.~Leung, W.~Matthews, and A.~Winter, ``Zero-error channel
  capacity and simulation assisted by non-local correlations,'' \emph{{IEEE}
  Transactions on Information Theory}, vol.~57, no.~8, pp. 5509--5523, Aug.
  2011.

\bibitem{Duan-2016a}
R.~Duan and A.~Winter, ``No-signalling-assisted zero-error capacity of quantum
  channels and an information theoretic interpretation of the lov{\'{a}}sz
  number,'' \emph{{IEEE} Transactions on Information Theory}, vol.~62, no.~2,
  pp. 891--914, Feb. 2016.

\bibitem{Wang-2016a}
\BIBentryALTinterwordspacing
X.~Wang and R.~Duan, ``On the quantum no-signalling assisted zero-error
  classical simulation cost of non-commutative bipartite graphs,'' in
  \emph{2016 {IEEE} International Symposium on Information Theory
  ({ISIT})}.\hskip 1em plus 0.5em minus 0.4em\relax {IEEE}, Jul. 2016.
  [Online]. Available: \url{https://doi.org/10.1109/isit.2016.7541698}
\BIBentrySTDinterwordspacing

\bibitem{Fang-2020a}
\BIBentryALTinterwordspacing
K.~Fang, X.~Wang, M.~Tomamichel, and M.~Berta, ``Quantum channel simulation and
  the channel's smooth max-information,'' \emph{{IEEE} Transactions on
  Information Theory}, vol.~66, no.~4, pp. 2129--2140, Apr. 2020. [Online].
  Available: \url{https://doi.org/10.1109/tit.2019.2943858}
\BIBentrySTDinterwordspacing

\bibitem{Bennett-2002a}
C.~Bennett, P.~Shor, J.~Smolin, and A.~Thapliyal, ``Entanglement-assisted
  capacity of a quantum channel and the reverse shannon theorem,'' \emph{{IEEE}
  Transactions on Information Theory}, vol.~48, no.~10, pp. 2637--2655, Oct.
  2002.

\bibitem{Bennett-2014a}
\BIBentryALTinterwordspacing
C.~H. Bennett, I.~Devetak, A.~W. Harrow, P.~W. Shor, and A.~Winter, ``The
  quantum reverse shannon theorem and resource tradeoffs for simulating quantum
  channels,'' \emph{{IEEE} Transactions on Information Theory}, vol.~60, no.~5,
  pp. 2926--2959, May 2014. [Online]. Available:
  \url{https://doi.org/10.1109/tit.2014.2309968}
\BIBentrySTDinterwordspacing

\bibitem{Berta-2011a}
\BIBentryALTinterwordspacing
M.~Berta, M.~Christandl, and R.~Renner, ``The quantum reverse shannon theorem
  based on one-shot information theory,'' \emph{Communications in Mathematical
  Physics}, vol. 306, no.~3, pp. 579--615, Aug. 2011. [Online]. Available:
  \url{https://doi.org/10.1007/s00220-011-1309-7}
\BIBentrySTDinterwordspacing

\bibitem{Heinosaari-2019a}
T.~Heinosaari and O.~Kerppo, ``Communication of partial ignorance with
  qubits,'' \emph{Journal of Physics A: Mathematical and Theoretical}, vol.~52,
  no.~39, p. 395301, 2019.

\bibitem{Heinosaari-2020a}
T.~Heinosaari, O.~Kerppo, and L.~Lepp\"{a}j\"{a}rvi, ``Communication tasks in
  operational theories,'' \emph{Journal of Physics A: Mathematical and
  Theoretical}, vol.~53, no.~43, p. 435302, Oct. 2020.

\bibitem{Frenkel-2015a}
P.~E. Frenkel and M.~Weiner, ``Classical information storage in an n-level
  quantum system,'' \emph{Communications in Mathematical Physics}, vol. 340,
  no.~2, pp. 563--574, Sep. 2015.

\bibitem{Berta-2013a}
M.~Berta, F.~G. S.~L. Brandao, M.~Christandl, and S.~Wehner, ``Entanglement
  cost of quantum channels,'' \emph{{IEEE} Transactions on Information Theory},
  vol.~59, no.~10, pp. 6779--6795, Oct. 2013.

\bibitem{Wang-2018a}
X.~Wang and M.~M. Wilde, ``Exact entanglement cost of quantum states and
  channels under ppt-preserving operations,'' 2018.

\bibitem{Wilde-2018a}
\BIBentryALTinterwordspacing
M.~M. Wilde, ``Entanglement cost and quantum channel simulation,'' \emph{Phys.
  Rev. A}, vol.~98, p. 042338, Oct 2018. [Online]. Available:
  \url{https://link.aps.org/doi/10.1103/PhysRevA.98.042338}
\BIBentrySTDinterwordspacing

\bibitem{Gour-2021a}
\BIBentryALTinterwordspacing
G.~Gour and C.~M. Scandolo, ``Entanglement of a bipartite channel,''
  \emph{Phys. Rev. A}, vol. 103, p. 062422, Jun 2021. [Online]. Available:
  \url{https://link.aps.org/doi/10.1103/PhysRevA.103.062422}
\BIBentrySTDinterwordspacing

\bibitem{Frenkel-2021a}
P.~E. Frenkel and M.~Weiner, ``On entanglement assistance to a noiseless
  classical channel,'' 2021.

\bibitem{Dall'Arno-2017a}
\BIBentryALTinterwordspacing
M.~Dall'Arno, S.~Brandsen, A.~Tosini, F.~Buscemi, and V.~Vedral,
  ``No-hypersignaling principle,'' \emph{Phys. Rev. Lett.}, vol. 119, p.
  020401, Jul 2017. [Online]. Available:
  \url{https://link.aps.org/doi/10.1103/PhysRevLett.119.020401}
\BIBentrySTDinterwordspacing

\bibitem{Doolittle-2021a}
B.~Doolittle and E.~Chitambar, ``Certifying the classical simulation cost of a
  quantum channel,'' 2021.

\bibitem{Bennett-1992a}
\BIBentryALTinterwordspacing
C.~H. Bennett and S.~J. Wiesner, ``Communication via one- and two-particle
  operators on einstein-podolsky-rosen states,'' \emph{Phys. Rev. Lett.},
  vol.~69, pp. 2881--2884, Nov 1992. [Online]. Available:
  \url{https://link.aps.org/doi/10.1103/PhysRevLett.69.2881}
\BIBentrySTDinterwordspacing

\bibitem{Konig-2009a}
R.~Konig, R.~Renner, and C.~Schaffner, ``The operational meaning of min- and
  max-entropy,'' \emph{{IEEE} Transactions on Information Theory}, vol.~55,
  no.~9, pp. 4337--4347, Sep. 2009.

\bibitem{Davies-1978a}
\BIBentryALTinterwordspacing
E.~Davies, ``Information and quantum measurement,'' \emph{{IEEE} Transactions
  on Information Theory}, vol.~24, no.~5, pp. 596--599, Sep. 1978. [Online].
  Available: \url{https://doi.org/10.1109/tit.1978.1055941}
\BIBentrySTDinterwordspacing

\bibitem{Helstrom-1976a}
C.~W. Helstrom, \emph{Quantum detection and estimation theory}.\hskip 1em plus
  0.5em minus 0.4em\relax Academic Press, New York, 1976.

\bibitem{Gurvits-2003a}
\BIBentryALTinterwordspacing
L.~Gurvits, ``Classical deterministic complexity of edmonds' problem and
  quantum entanglement,'' in \emph{Proceedings of the thirty-fifth {ACM}
  symposium on Theory of computing - {STOC} '03}.\hskip 1em plus 0.5em minus
  0.4em\relax {ACM} Press, 2003. [Online]. Available:
  \url{https://doi.org/10.1145/780542.780545}
\BIBentrySTDinterwordspacing

\bibitem{Peres-1996a}
\BIBentryALTinterwordspacing
A.~Peres, ``Separability criterion for density matrices,'' \emph{Phys. Rev.
  Lett.}, vol.~77, pp. 1413--1415, Aug 1996. [Online]. Available:
  \url{https://link.aps.org/doi/10.1103/PhysRevLett.77.1413}
\BIBentrySTDinterwordspacing

\bibitem{Horodecki-1996a}
\BIBentryALTinterwordspacing
M.~Horodecki, P.~Horodecki, and R.~Horodecki, ``Separability of mixed states:
  necessary and sufficient conditions,'' \emph{Physics Letters A}, vol. 223,
  no. 1-2, pp. 1--8, Nov. 1996. [Online]. Available:
  \url{https://doi.org/10.1016/s0375-9601(96)00706-2}
\BIBentrySTDinterwordspacing

\bibitem{Tomamichel-2015}
M.~Tomamichel, \emph{Quantum information processing with finite resources:
  mathematical foundations}.\hskip 1em plus 0.5em minus 0.4em\relax Springer,
  2015, vol.~5.

\bibitem{Wilde-2014a}
\BIBentryALTinterwordspacing
M.~M. Wilde, A.~Winter, and D.~Yang, ``Strong converse for the classical
  capacity of entanglement-breaking and hadamard channels via a sandwiched
  r{\'{e}}nyi relative entropy,'' \emph{Communications in Mathematical
  Physics}, vol. 331, no.~2, pp. 593--622, Jul. 2014. [Online]. Available:
  \url{https://doi.org/10.1007/s00220-014-2122-x}
\BIBentrySTDinterwordspacing

\bibitem{Beigi-2013a}
S.~Beigi, ``Sandwiched r{\'{e}}nyi divergence satisfies data processing
  inequality,'' \emph{Journal of Mathematical Physics}, vol.~54, no.~12, p.
  122202, Dec. 2013.

\bibitem{Datta-2009a}
N.~Datta, ``Min- and max-relative entropies and a new entanglement monotone,''
  \emph{{IEEE} Transactions on Information Theory}, vol.~55, no.~6, pp.
  2816--2826, Jun. 2009.

\bibitem{Bennett-1996a}
\BIBentryALTinterwordspacing
C.~H. Bennett, D.~P. DiVincenzo, J.~A. Smolin, and W.~K. Wootters,
  ``Mixed-state entanglement and quantum error correction,'' \emph{Phys. Rev.
  A}, vol.~54, pp. 3824--3851, Nov 1996. [Online]. Available:
  \url{https://link.aps.org/doi/10.1103/PhysRevA.54.3824}
\BIBentrySTDinterwordspacing

\bibitem{Shimony-1995a}
A.~Shimony, ``Degree of entanglement,'' \emph{Annals of the New York Academy of
  Sciences}, vol. 755, no.~1, pp. 675--679, Apr. 1995.

\bibitem{Wei-2003a}
T.-C. Wei and P.~M. Goldbart, ``Geometric measure of entanglement and
  applications to bipartite and multipartite quantum states,'' \emph{Phys. Rev.
  A}, vol.~68, p. 042307, Oct 2003.

\bibitem{Werner-2002a}
R.~F. Werner and A.~S. Holevo, ``Counterexample to an additivity conjecture for
  output purity of quantum channels,'' \emph{Journal of Mathematical Physics},
  vol.~43, no.~9, pp. 4353--4357, Sep. 2002.

\bibitem{Zhu-2011a}
H.~Zhu, L.~Chen, and M.~Hayashi, ``Additivity and non-additivity of
  multipartite entanglement measures,'' \emph{New Journal of Physics}, vol.~13,
  no.~1, p. 019501, Jan. 2011.

\bibitem{Jung-2008}
E.~Jung, M.-R. Hwang, H.~Kim, M.-S. Kim, D.~Park, J.-W. Son, and S.~Tamaryan,
  ``Reduced state uniquely defines the groverian measure of the original pure
  state,'' \emph{Physical Review A}, vol.~77, no.~6, p. 062317, 2008.

\bibitem{Leung-2015a}
D.~Leung and W.~Matthews, ``On the power of ppt-preserving and non-signalling
  codes,'' \emph{IEEE Transactions on Information Theory}, vol.~61, no.~8, pp.
  4486--4499, 2015.

\bibitem{Werner-1989a}
\BIBentryALTinterwordspacing
R.~F. Werner, ``Quantum states with einstein-podolsky-rosen correlations
  admitting a hidden-variable model,'' \emph{Physical Review A}, vol.~40,
  no.~8, pp. 4277--4281, Oct. 1989. [Online]. Available:
  \url{https://doi.org/10.1103/physreva.40.4277}
\BIBentrySTDinterwordspacing

\bibitem{Vollbrecht-2001a}
K.~G.~H. Vollbrecht and R.~F. Werner, ``Entanglement measures under symmetry,''
  \emph{Physical Review A}, vol.~64, no.~6, p. 062307, 2001.

\bibitem{Horodecki-1999a}
M.~Horodecki, P.~Horodecki, and R.~Horodecki, ``General teleportation channel,
  singlet fraction, and quasidistillation,'' \emph{Phys. Rev. A}, vol.~60, pp.
  1888--1898, Sep 1999.

\bibitem{Christandl-2012a}
M.~Christandl, N.~Schuch, and A.~Winter, ``Entanglement of the antisymmetric
  state,'' \emph{Communications in Mathematical Physics}, vol. 311, no.~2, pp.
  397--422, Mar. 2012.

\bibitem{Hubener-2009a}
R.~H\"ubener, M.~Kleinmann, T.-C. Wei, C.~Gonz\'alez-Guill\'en, and O.~G\"uhne,
  ``Geometric measure of entanglement for symmetric states,'' \emph{Phys. Rev.
  A}, vol.~80, p. 032324, Sep 2009.

\bibitem{Nielsen-1999a}
M.~A. Nielsen, ``Conditions for a class of entanglement transformations,''
  \emph{Phys. Rev. Lett.}, vol.~83, pp. 436--439, Jul 1999.

\bibitem{Lo-2001a}
H.-K. Lo and S.~Popescu, ``Concentrating entanglement by local actions: Beyond
  mean values,'' \emph{Phys. Rev. A}, vol.~63, p. 022301, Jan 2001.

\bibitem{Doherty-2004}
A.~C. Doherty, P.~A. Parrilo, and F.~M. Spedalieri, ``Complete family of
  separability criteria,'' \emph{Physical Review A}, vol.~69, no.~2, p. 022308,
  2004.

\bibitem{King-2002}
C.~King, ``\BIBforeignlanguage{English}{Maximal p-norms of entanglement
  breaking channels},'' \emph{\BIBforeignlanguage{English}{Quantum Information
  and Computation}}, vol.~3, no.~2, pp. 186--190, 2003.

\bibitem{CVChannel2021}
\BIBentryALTinterwordspacing
B.~Doolittle and I.~George, ``Cvchannel.jl,'' Sep. 2021. [Online]. Available:
  \url{https://github.com/ChitambarLab/CVChannel.jl}
\BIBentrySTDinterwordspacing

\bibitem{bezanson2017julia}
\BIBentryALTinterwordspacing
J.~Bezanson, A.~Edelman, S.~Karpinski, and V.~B. Shah, ``Julia: A fresh
  approach to numerical computing,'' \emph{SIAM review}, vol.~59, no.~1, pp.
  65--98, 2017. [Online]. Available: \url{https://doi.org/10.1137/141000671}
\BIBentrySTDinterwordspacing

\bibitem{convexjl2014}
M.~Udell, K.~Mohan, D.~Zeng, J.~Hong, S.~Diamond, and S.~Boyd, ``Convex
  optimization in {J}ulia,'' \emph{SC14 Workshop on High Performance Technical
  Computing in Dynamic Languages}, 2014.

\bibitem{scs2019}
B.~O'Donoghue, E.~Chu, N.~Parikh, and S.~Boyd, ``{SCS}: Splitting conic solver,
  version 2.1.4,'' \url{https://github.com/cvxgrp/scs}, Nov. 2019.

\bibitem{Reimpell2005}
\BIBentryALTinterwordspacing
M.~Reimpell and R.~F. Werner, ``Iterative optimization of quantum error
  correcting codes,'' \emph{Physical Review Letters}, vol.~94, no.~8, Mar.
  2005. [Online]. Available:
  \url{https://doi.org/10.1103/physrevlett.94.080501}
\BIBentrySTDinterwordspacing

\bibitem{Kosut2009}
\BIBentryALTinterwordspacing
R.~L. Kosut and D.~A. Lidar, ``Quantum error correction via convex
  optimization,'' \emph{Quantum Information Processing}, vol.~8, no.~5, pp.
  443--459, Jul. 2009. [Online]. Available:
  \url{https://doi.org/10.1007/s11128-009-0120-2}
\BIBentrySTDinterwordspacing

\bibitem{Leditzky-2018}
F.~Leditzky, D.~Leung, and G.~Smith, ``Dephrasure channel and superadditivity
  of coherent information,'' \emph{Physical review letters}, vol. 121, no.~16,
  p. 160501, 2018.

\bibitem{Siddhu-2020}
V.~Siddhu, ``Log-singularities for studying capacities of quantum channels,''
  \emph{arXiv preprint arXiv:2003.10367}, 2020.

\bibitem{Holevo-1998a}
\BIBentryALTinterwordspacing
A.~Holevo, ``The capacity of the quantum channel with general signal states,''
  \emph{{IEEE} Transactions on Information Theory}, vol.~44, no.~1, pp.
  269--273, 1998. [Online]. Available: \url{https://doi.org/10.1109/18.651037}
\BIBentrySTDinterwordspacing

\bibitem{Schumacher-1997a}
B.~Schumacher and M.~D. Westmoreland, ``Sending classical information via noisy
  quantum channels,'' \emph{Phys. Rev. A}, vol.~56, pp. 131--138, Jul 1997.

\bibitem{Holevo-2002a}
A.~S. Holevo, ``On entanglement-assisted classical capacity,'' \emph{Journal of
  Mathematical Physics}, vol.~43, no.~9, pp. 4326--4333, 2002.

\bibitem{Watrous-2018a}
J.~Watrous, \emph{The Theory of Quantum Information}.\hskip 1em plus 0.5em
  minus 0.4em\relax Cambridge university press, 2018.

\end{thebibliography}

\end{document}